\documentclass[10pt,twocolumn,conference]{IEEEtran}

\usepackage{graphicx}
\usepackage{listings,color}
\usepackage{listings}

\ttfamily
\DeclareFontShape{OT1}{lmtt}{m}{it}
     {<->sub*lmtt/m/sl}{}
\lstdefinestyle{default}{language=Anglican, basicstyle=\ttfamily, columns=flexible, showstringspaces=false}

\usepackage{listings}
\usepackage{color}
\usepackage{xcolor}
\definecolor{blue}{rgb}{0,0.3,0.7}
\definecolor{red}{rgb}{0.60,0.0,0.0}
\definecolor{purple}{rgb}{0.5,0,0.7}
\definecolor{cyan}{rgb}{0.0,0.6,0.5}
\definecolor{gray}{rgb}{0.4,0.4,0.4}

\lstdefinelanguage{scheme}
{sensitive, %
 alsoletter={:,-,+,*,?,/,!,>,<}, %
 morecomment=[l];, %
}[comments]

\lstdefinelanguage{anglican}%
{
 morekeywords=[1]{},
 morekeywords=[2]{%
   def, def-, defn, defn-, defmacro, defmulti, defmethod, %
   defstruct, defonce, declare, definline, definterface, %
   defprotocol, defrecord, defstruct, deftype, defproject, ns, %
 }, %
 morekeywords=[3]{->, ->>, .., amap, and, areduce, as->, assert, binding, %
   bound-fn, case, comment, cond, cond->, cond->>, condp, declare, definline, %
   definterface, defmacro, defmethod, defmulti, defn, defn-, defonce, %
   defprotocol, defrecord, defstruct, deftype, delay, doseq, dosync, dotimes, %
   doto, extend-protocol, extend-type, fn, for, future, gen-class, %
   gen-interface, if, if-let, if-not, if-some, import, io!, lazy-cat, lazy-seq, let, %
   letfn, locking, loop, memfn, ns, or, proxy, proxy-super, pvalues, %
   recur, refer-clojure, reify, some->, some->>, sync, time, when, when-first, %
   when-let, when-not, when-some, while, with-bindings, with-in-str, %
   with-loading-context, with-local-vars, with-open, with-out-str, %
   with-precision, with-redefs}, %
  morekeywords=[4]{*, *', +, +', -, -', ->ArrayChunk, ->Vec, ->VecNode, %
    ->VecSeq, -cache-protocol-fn, -reset-methods, /, <, <=, =, ==, >, >=, %
    accessor, aclone, add-classpath, add-watch, agent, agent-error, %
    agent-errors, aget, alength, alias, all-ns, alter, alter-meta!, %
    alter-var-root, ancestors, apply, array-map, aset, aset-boolean, aset-byte, %
    aset-char, aset-double, aset-float, aset-int, aset-long, aset-short, assoc, %
    assoc!, assoc-in, associative?, atom, await, await-for, await1, bases, bean, %
    bigdec, bigint, biginteger, bit-and, bit-and-not, bit-clear, bit-flip, %
    bit-not, bit-or, bit-set, bit-shift-left, bit-shift-right, bit-test, %
    bit-xor, boolean, boolean-array, booleans, bound-fn*, bound?, butlast, byte, %
    byte-array, bytes, cast, char, char-array, char?, chars, chunk, %
    chunk-append, chunk-buffer, chunk-cons, chunk-first, chunk-next, chunk-rest, %
    chunked-seq?, class, class?, clear-agent-errors, clojure-version, coll?, %
    commute, comp, comparator, compare, compare-and-set!, compile, complement, %
    concat, conj, conj!, cons, constantly, construct-proxy, contains?, count, %
    counted?, create-ns, create-struct, cycle, dec, dec', decimal?, delay?, %
    deliver, denominator, deref, derive, descendants, destructure, disj, disj!, %
    dissoc, dissoc!, distinct, distinct?, doall, dorun, double, double-array, %
    doubles, drop, drop-last, drop-while, empty, empty?, ensure, %
    enumeration-seq, error-handler, error-mode, eval, even?, every-pred, every?, %
    ex-data, ex-info, extend, extenders, extends?, false?, ffirst, file-seq, %
    filter, filter-ns-publics, filterv, find, find-keyword, find-ns, %
    find-protocol-impl, find-protocol-method, find-var, first, flatten, float, %
    float-array, float?, floats, flush, fn?, fnext, fnil, force, format, %
    frequencies, future-call, future-cancel, future-cancelled?, future-done?, %
    future?, gensym, get, get-in, get-method, get-proxy-class, %
    get-thread-bindings, get-validator, group-by, hash, hash-combine, hash-map, %
    hash-ordered-coll, hash-set, hash-unordered-coll, identical?, identity, %
    ifn?, in-ns, inc, inc', init-proxy, instance?, int, int-array, integer?, %
    interleave, intern, interpose, into, into-array, ints, isa?, iterate, %
    iterator-seq, juxt, keep, keep-indexed, key, keys, keyword, keyword?, last, %
    line-seq, list, list*, list?, load, load-file, load-reader, load-string, %
    loaded-libs, long, long-array, longs, macroexpand, macroexpand-1, %
    make-array, make-hierarchy, map, map-indexed, map?, mapcat, mapv, max, %
    max-key, memoize, merge, merge-with, meta, method-sig, methods, min, %
    min-key, mix-collection-hash, mod, munge, name, namespace, namespace-munge, %
    neg?, newline, next, nfirst, nil?, nnext, not, not-any?, not-empty, %
    not-every?, not=, ns-aliases, ns-functions, ns-imports, ns-interns, %
    ns-macros, ns-map, ns-name, ns-publics, ns-refers, ns-resolve, ns-unalias, %
    ns-unmap, nth, nthnext, nthrest, num, number?, numerator, object-array, %
    odd?, parents, partial, partition, partition-all, partition-by, pcalls, %
    peek, persistent!, pmap, pop, pop!, pop-thread-bindings, pos?, pr, pr-str, %
    prefer-method, prefers, print, print-ctor, print-simple, print-str, printf, %
    println, println-str, prn, prn-str, promise, proxy-call-with-super, %
    proxy-mappings, proxy-name, push-thread-bindings, quot, rand, rand-int, %
    rand-nth, range, ratio?, rational?, rationalize, re-find, re-groups, %
    re-matcher, re-matches, re-pattern, re-seq, read, read-line, read-string, %
    realized?, record?, reduce, reduce-kv, reduced, reduced?, reductions, ref, %
    ref-history-count, ref-max-history, ref-min-history, ref-set, refer, %
    release-pending-sends, rem, remove, remove-all-methods, remove-method, %
    remove-ns, remove-watch, repeat, repeatedly, replace, replicate, require, %
    reset!, reset-meta!, resolve, rest, restart-agent, resultset-seq, reverse, %
    reversible?, rseq, rsubseq, satisfies?, second, select-keys, send, send-off, %
    send-via, seq, seq?, seque, sequence, sequential?, set, %
    set-agent-send-executor!, set-agent-send-off-executor!, set-error-handler!, %
    set-error-mode!, set-validator!, set?, short, short-array, shorts, shuffle, %
    shutdown-agents, slurp, some, some-fn, some?, sort, sort-by, sorted-map, %
    sorted-map-by, sorted-set, sorted-set-by, sorted?, special-symbol?, spit, %
    split-at, split-with, str, string?, struct, struct-map, subs, subseq, %
    subvec, supers, swap!, symbol, symbol?, take, take-last, take-nth, %
    take-while, test, the-ns, thread-bound?, to-array, to-array-2d, trampoline, %
    transient, tree-seq, true?, type, unchecked-add, unchecked-add-int, %
    unchecked-byte, unchecked-char, unchecked-dec, unchecked-dec-int, %
    unchecked-divide-int, unchecked-double, unchecked-float, unchecked-inc, %
    unchecked-inc-int, unchecked-int, unchecked-long, unchecked-multiply, %
    unchecked-multiply-int, unchecked-negate, unchecked-negate-int, %
    unchecked-remainder-int, unchecked-short, unchecked-subtract, %
    unchecked-subtract-int, underive, unsigned-bit-shift-right, update-in, %
    update-proxy, use, val, vals, var-get, var-set, var?, vary-meta, vec, %
    vector, vector-of, vector?, with-bindings*, with-meta, with-redefs-fn, %
    xml-seq, zero?, zipmap}, %
  morekeywords=[5]{def-cps-fn, defanglican, defm, defquery, defun, defproc, defdist}, %
  morekeywords=[6]{cps-fn, fm, lambda, mem, query, with-primitive-procedures}, %
  morekeywords=[7]{%
    doquery, %
    conditional, %
    collect-by, equalize, exec, infer, log-marginal, print-predicts, %
    rand, rand-int, rand-nth, rand-roulette, stripdown, warmup, %
    ->CRP-process, ->DP-process, ->GP-process, %
    ->bernoulli-distribution, ->beta-distribution, ->binomial-distribution, %
    ->categorical-crp-distribution, ->categorical-distribution, %
    ->categorical-dp-distribution, ->chi-squared-distribution, %
    ->dirichlet-distribution, ->discrete-distribution, %
    ->exponential-distribution, ->flip-distribution, ->gamma-distribution, %
    ->mvn-distribution, ->normal-distribution, ->poisson-distribution, %
    ->sample, ->observe, sample*, observe*, %
    ->uniform-continuous-distribution, ->uniform-discrete-distribution, %
    ->wishart-distribution, CRP, DP, GP, abs, absorb, acos, asin, atan, %
    bernoulli, beta, binomial, categorical, categorical-crp, categorical-dp, %
    cbrt, ceil, chi-squared, cos, cosh, cov, dirichlet, discrete, exp, %
    exponential, flip, floor, gamma, gen-matrix, log, log-gamma-fn, %
    log-mv-gamma-fn, log-sum-exp, map->CRP-process, map->DP-process, %
    map->GP-process, map->bernoulli-distribution, map->beta-distribution, %
    map->binomial-distribution, map->categorical-crp-distribution, %
    map->categorical-distribution, map->categorical-dp-distribution, %
    map->chi-squared-distribution, map->dirichlet-distribution, %
    map->discrete-distribution, map->exponential-distribution, %
    map->flip-distribution, map->gamma-distribution, map->mvn-distribution, %
    map->normal-distribution, map->poisson-distribution, %
    map->uniform-continuous-distribution, map->uniform-discrete-distribution, %
    map->wishart-distribution, mvn, normal, poisson, pow, produce, %
    rint, round, signum, sin, sinh, sqrt, tag, tan, tanh, transform-sample, %
    uniform-continuous, uniform-discrete, wishart, %
    add-log-weight, add-predict, clear-predicts, get-log-weight, %
    get-mem, get-predicts, in-mem?, set-log-weight, set-mem, %
  }, %
  morekeywords=[8]{factor, observe, predict, retrieve, sample, store}, %
  sensitive, %
  alsoletter={:,-,+,*,?,/,!,>,<}, %
  morecomment=[l];, %
  morestring=[b]", %
  keywordsprefix=:, %
}[keywords,comments,strings]
 
\newcommand{\commentout}[1]{}
\newcommand{\shortversion}[1]{#1}
\newcommand{\longversion}[1]{}

\newcommand{\hy}[1]{
\smallskip\noindent\fbox{\begin{minipage}{.97\linewidth}{\bf HY:}
{\rm #1}\end{minipage}}}
\newcommand{\sss}[1]{
\smallskip\noindent\fbox{\begin{minipage}{.97\linewidth}{\bf SS:}
{\rm #1}\end{minipage}}}
\newcommand{\ok}[1]{
\smallskip\noindent\fbox{\begin{minipage}{.97\linewidth}{\bf OK:}
{\rm #1}\end{minipage}}}
\newcommand{\ch}[1]{
\smallskip\noindent\fbox{\begin{minipage}{.97\linewidth}{\bf CH:}
{\rm #1}\end{minipage}}}

\renewcommand{\hy}[1]{}
\renewcommand{\sss}[1]{}
\renewcommand{\ok}[1]{}
\renewcommand{\ch}[1]{}

\newcommand{\XI}{X^\omega}
\newcommand{\Xn}{X^n}
\newcommand{\RR}{\mathbb{R}}
\newcommand{\NN}{\mathbb{N}}
\newcommand{\sigalg}[1]{\Sigma_{#1}}
\newcommand{\qbtosig}[1]{\Sigma_{#1}}
\newcommand{\qb}[1]{M_{#1}}
\newcommand{\sigtoqb}[1]{M_{#1}}
\newcommand{\inv}[1]{{#1}^{\textup{-1}}}
\newcommand{\id}[1]{\mathrm{id}_{#1}}

\newcommand{\iidn}{\mathrm{iid}_n}
\newcommand{\iotan}{\iota_n}
\newcommand{\inviotan}{\inv{\iota}_n}
\newcommand{\op}[1]{{#1}^{\mathrm{op}}}
\newcommand{\defeq}{\stackrel {\textup{def}}=}
\newcommand{\QBS}{\mathbf{QBS}}
\newcommand{\Meas}{\mathbf{Meas}}
\newcommand{\Set}{\mathbf{Set}}
\newcommand{\SMeas}{\mathbf{SMeas}}
\newcommand{\dd}{\mathrm{d}}
\newcommand{\denot}[1]{\llbracket#1\rrbracket}
\newcommand{\Pmonad}{P}
\newcommand{\Giry}{G}
\newcommand{\curry}{\mathsf{curry}}
\newcommand{\uncurry}{\mathsf{uncurry}}
\newcommand{\bindsymbol}{\scalebox{0.5}[1]{$>\!>=$}}
\newcommand{\bind}[2]{#1\mathrel{\bindsymbol} #2}
\newcommand{\bindname}{(\bindsymbol)}
\newcommand{\bindGsymbol}{\scalebox{0.5}[1]{$>\!>=$}_G}
\newcommand{\bindG}[2]{#1\mathrel{\bindGsymbol} #2}

\newcommand{\munit}{\eta}

\ifCLASSINFOpdf
\else
\fi
\usepackage{amssymb,amsthm,amsmath,amsfonts,mathtools,stmaryrd}
\usepackage{xy}
\xyoption{all}
\newtheorem{definition}{Definition}
\newtheorem{proposition}[definition]{Proposition}
\newtheorem{theorem}[definition]{Theorem}
\newtheorem{lemma}[definition]{Lemma}
\newtheorem{example}[definition]{Example}
\begin{document}
%
\title{A Convenient Category for \\ Higher-Order Probability Theory}

\author{
\IEEEauthorblockN{Chris Heunen}
\IEEEauthorblockA{University of Edinburgh, UK}
\and
\IEEEauthorblockN{Ohad Kammar}
\IEEEauthorblockA{University of Oxford, UK}
\and
\IEEEauthorblockN{Sam Staton}
\IEEEauthorblockA{University of Oxford, UK}
\and
\IEEEauthorblockN{Hongseok Yang}
\IEEEauthorblockA{University of Oxford, UK}
}


%

\IEEEoverridecommandlockouts
\IEEEpubid{\makebox[\columnwidth]{978-1-5090-3018-7/17/\$31.00~
\copyright2017 IEEE \hfill} \hspace{\columnsep}\makebox[\columnwidth]{ }}

\maketitle

\begin{abstract}
Higher-order probabilistic programming languages allow programmers to write sophisticated models in machine learning and statistics in a succinct and structured way, 
but step outside the standard measure-theoretic formalization of probability theory.
Programs may use both higher-order functions and continuous distributions, or even define a probability distribution on functions.
But standard probability theory does not handle higher-order functions well: the category of measurable spaces is not cartesian closed.

Here we introduce quasi-Borel spaces. 
We show that these spaces:
form a new formalization of probability theory replacing measurable spaces;
form a cartesian closed category and so support higher-order functions;
form a well-pointed category and so support good proof principles for equational reasoning;
and support continuous probability distributions.
We demonstrate the use of quasi-Borel spaces for higher-order functions and probability by: showing that a well-known construction of probability theory involving random functions gains a cleaner expression; and generalizing de Finetti's theorem, that is a crucial theorem in probability theory, to quasi-Borel spaces.

\end{abstract}

\IEEEpeerreviewmaketitle
\allowdisplaybreaks

\section{Introduction}\label{sec:intro}
To express probabilistic models in machine learning and statistics in a succinct and structured way, it pays to use \emph{higher-order} programming languages, such as Church~\cite{goodman_uai_2008}, Venture~\cite{Mansinghka-venture14}, or Anglican~\cite{wood-aistats-2014}. These languages support advanced features from both
programming language theory and probability theory, while
providing generic inference algorithms for answering probabilistic queries, such as marginalization and posterior computation, for all models written in the language. As a result, the programmer can succinctly express a sophisticated probabilistic model and explore its properties while avoiding the nontrivial busywork of designing a custom inference algorithm.
%
%

This exciting development comes at a foundational price.
Programs in these languages may combine higher-order functions and continuous
distributions, or even define a probability distribution on functions.
But the standard measure-theoretic formalization of probability theory does not handle higher-order functions well, as the category of measurable spaces is not cartesian closed~\cite{aumann:functionspaces}. 
For instance, the Anglican implementation of Bayesian linear regression in Figure~\ref{fig:linearregression} goes beyond the standard measure-theoretic foundation of probability theory, as it defines a probability distribution on functions $\RR \to \RR$. 

\makeatletter
\lst@Key{countblanklines}{true}[t]%
    {\lstKV@SetIf{#1}\lst@ifcountblanklines}

\lst@AddToHook{OnEmptyLine}{%
    \lst@ifnumberblanklines\else%
       \lst@ifcountblanklines\else%
         \advance\c@lstnumber-\@ne\relax%
       \fi%
    \fi}
\makeatother
\lstset{frame=single, 
 xleftmargin=8.25mm, framexleftmargin=7mm, numbers=left, numberblanklines=false, linewidth=.48\textwidth
}
\begin{figure}
%
\begin{lstlisting}[style=default,countblanklines=false, basicstyle=\ttfamily\small,escapechar=\|]
(defquery Bayesian-linear-regression |%
\vskip-.6\baselineskip|

  (let [f (let [s (sample (normal 0.0 3.0))
                  b (sample (normal 0.0 3.0))]
              (fn [x] (+ (* s x) b)))] |%
\vskip-.6\baselineskip|

    (observe (normal (f 1.0) 0.5) 2.5)        
    (observe (normal (f 2.0) 0.5) 3.8)
    (observe (normal (f 3.0) 0.5) 4.5)
    (observe (normal (f 4.0) 0.5) 6.2)
    (observe (normal (f 5.0) 0.5) 8.0)  |%
\vskip-.6\baselineskip|

    (predict :f f)))
\end{lstlisting}
\fbox{\includegraphics[width=.97\linewidth]{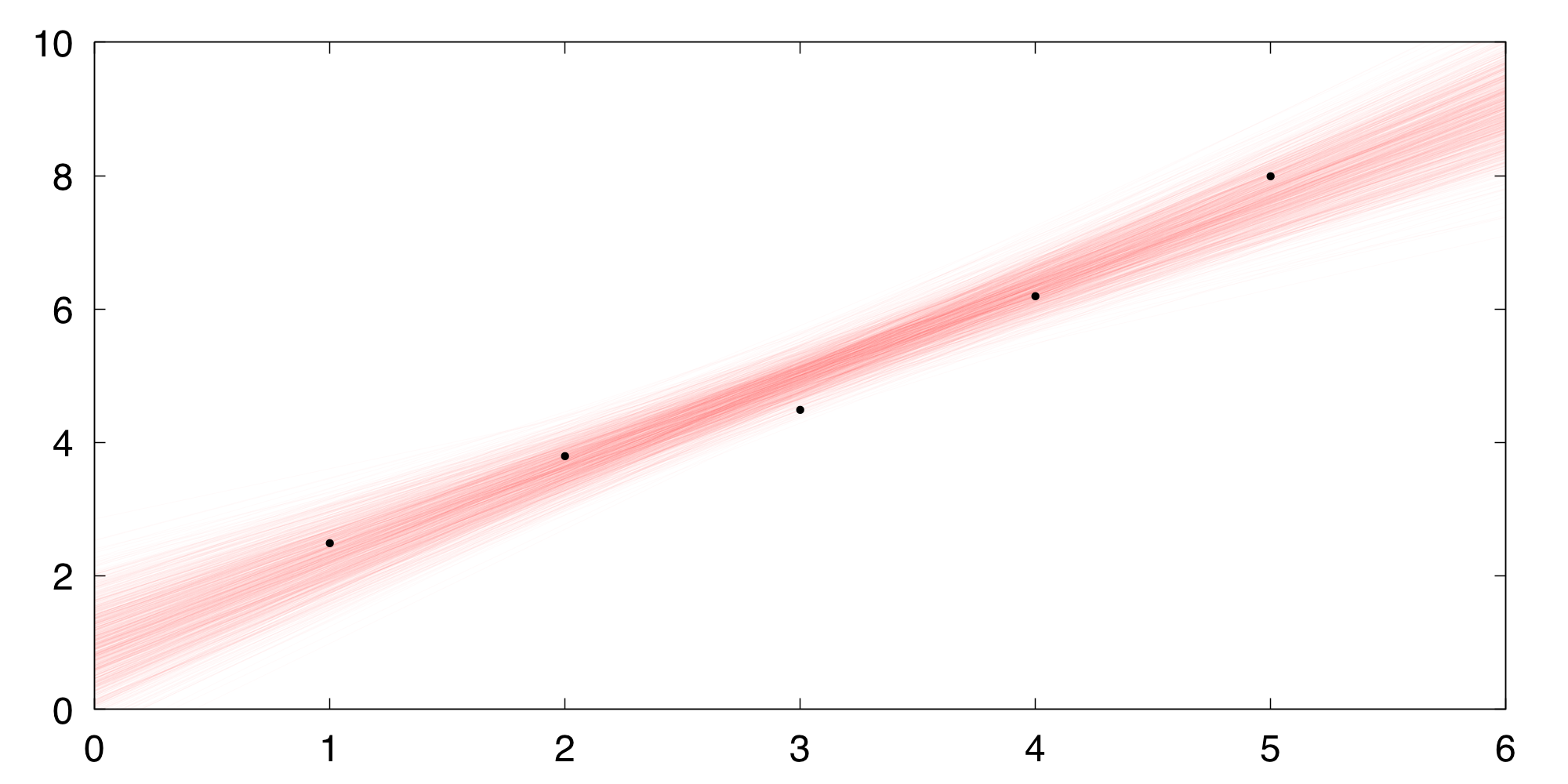}}
\caption{Bayesian linear regression in Anglican. The program defines a probability
distribution on functions $\RR \to \RR$. It first samples
a random linear function \texttt{f} by randomly selecting
slope \texttt{s} and intercept \texttt{b}. It then adjusts the probability distribution of the function
to better describe five observations $(1.0,2.5)$, $(2.0,3.8)$, $(3.0,4.5)$, $(4.0,6.2)$
and $(5.0,8.0)$ by posterior computation.
In the graph, 
each line has been sampled from the posterior distribution over linear functions.}
\label{fig:linearregression}
\end{figure}

We introduce a new formalization of probability theory that accommodates higher-order functions. The main notion replacing a measurable space is a \emph{quasi-Borel space}: a set~$X$ equipped with a collection of functions $\qb X \subseteq {[\RR \to X]}$ satisfying
certain conditions (Def.~\ref{def:qbs}). Intuitively, $\qb X$ is the set of random variables of type $X$.
Here $\RR$ means that the randomness of random variables in $\qb X$
comes from (a probability distribution on) $\RR$, one of the best behaving measurable spaces. 
Thus the primitive notion shifts from measurable subset to random variable, which is traditionally a derived notion.
For related ideas see \S\ref{sec:related}.

Quasi-Borel spaces have good properties and structure.
\begin{itemize}
\item The category of quasi-Borel spaces is \emph{well-pointed}, since a morphism is just a structure-preserving function (\S\ref{sec:quasiborel}). (This is in contrast to~\cite[\S8]{statonyangheunenkammarwood:higherorder}).

\item The category of quasi-Borel spaces is cartesian closed (\S\ref{sec:structure}), so that it becomes a setting to study probability distributions on \emph{higher-order} functions. 

\item There is a natural notion of probability measure on quasi-Borel spaces (Def.~\ref{def:probabilitymeasure}).
The space of all probability measures is again a 
quasi-Borel space, and forms the basis for a commutative \emph{monad} on the category of quasi-Borel spaces (\S\ref{sec:giry}). Thus quasi-Borel spaces form semantics for a probabilistic programming language in the monadic style~\cite{moggi-monads}.
\end{itemize}
We also illustrate the use of quasi-Borel spaces.
\begin{itemize}
\item \emph{Bayesian regression} (\S\ref{sec:example}). 
  Quasi-Borel spaces are a natural setting for understanding programs such as the one in Figure~\ref{fig:linearregression}:
  the prior (Lines 2--4) defines a probability distribution over functions~\lstinline|f|, i.e.\ a measure on $\RR^\RR$, and the posterior (illustrated in the graph), is again a probability measure on $\RR^\RR$, conditioned by the observations (Lines 5--9). 
\item \emph{Randomization} (\S\ref{sec:functions}). 
A key idea of categorical logic is that $\forall\exists$ statements should become statements
about quotients of objects. 
The structure of quasi-Borel spaces allows us to rephrase a crucial randomization lemma in this way. 
Classically, it says that every probability kernel arises from a random function. 
In the setting of quasi-Borel spaces, it says that
the space of probability kernels $P(\RR)^X$ is a quotient of the space of random functions, $P(\RR^X)$ (Theorem~\ref{theorem:random-quotient}).
Notice that the higher-order structure of quasi-Borel spaces allows us to succinctly state this result.
\item \emph{De Finetti's theorem} (\S\ref{sec:definetti}). 
Probability theorists often encounter problems when working with arbitrary probability measures on arbitrary measurable spaces.
Quasi-Borel spaces allow us to better manage the source of randomness. 
For example, de Finetti's theorem is a foundational result in Bayesian statistics which says that every exchangeable random sequence can be generated
by randomly mixing multiple independent and identically distributed sequences.
The theorem is known to hold for standard Borel spaces~\cite{deFinetti37} 
or measurable spaces that arise from good topologies~\cite{HewittSavage55}, but not for arbitrary measurable spaces~\cite{Dubins1979}. We show that it holds for all quasi-Borel spaces (Theorem~\ref{thm:deFinetti-qbs}). 
\end{itemize}
All of this is evidence that quasi-Borel spaces form a convenient category for higher-order probability theory.



\section{Preliminaries on probability measures and measurable spaces}\label{sec:prelims}
\begin{definition}\label{def:borel}
  The \emph{Borel sets} form the least collection $\Sigma_\RR$ of subsets of $\RR$ that
  satisfies the following properties:
  \begin{itemize}
   \item intervals $(a,b)$ are Borel sets;
   \item complements of Borel sets are Borel;
   \item countable unions of Borel sets are Borel.
  \end{itemize}
\end{definition}

The Borel sets play a crucial role in probability theory because of the tight connection
between the notion of probability measure and the axiomatization of Borel sets. 

\begin{definition}\label{def:borelprob}
  A \emph{probability measure} on $\RR$ is a function $\mu\colon \sigalg\RR\to [0,1]$ satisfying $\mu(\RR)=1$ and $\mu(\biguplus S_i)=\sum \mu(S_i)$ for any countable sequence of disjoint Borel sets $S_i$.
\end{definition}

The natural generalization gives measurable spaces.

\begin{definition}\label{def:meas-space}
  A \emph{$\sigma$-algebra} on a set $X$ is a nonempty family of subsets of $X$ that is closed under complements and countable unions. A \emph{measurable space} is a pair $(X,\sigalg X)$ of a set $X$ and a $\sigma$-algebra $\sigalg X$ on it. A probability measure on a measurable space $X$ is a function $\mu\colon \sigalg X\to [0,1]$ satisfying $\mu(X)=1$ and $\mu(\biguplus S_i)=\sum \mu(S_i)$ for any countable sequence of disjoint sets $S_i\in\sigalg X$.
\end{definition}

The Borel sets of the reals form a leading example of a $\sigma$-algebra.
Other important examples are countable sets with their \emph{discrete $\sigma$-algebra}, which contains all subsets.
We can characterize these spaces as standard Borel spaces, but first introduce the appropriate structure-preserving maps.

\begin{definition}\label{def:measurablefunction}
  Let $(X,\sigalg X)$ and $(Y,\sigalg Y)$ be measurable spaces. 
  A \emph{measurable function} $f\colon X\to Y$ is a function such that $\inv{f}(U)\in\sigalg X$ when $U\in\sigalg Y$. 
\end{definition}

Thus a measurable function $f\colon X\to Y$ lets us \emph{push-forward} a probability measure $\mu$ on $X$ to a probability measure $f_*\mu$ on $Y$ by $(f_*\mu)(U)=\mu(\inv{f}(U))$.  Measurable spaces and measurable functions form a category $\Meas$.

Real-valued measurable functions $f \colon X \to \RR$ can be integrated with respect to a probability measure $\mu$ on $(X,\sigalg X)$. The \emph{integral} of a nonnegative function $f$ is
\[
\int_X f \,\dd\mu \defeq \sup_{\{U_i\}} \sum_i \left(\mu(U_i) \cdot \inf_{x \in U_i} f(x)\right)\text,
\]
where $\{U_i\}$ ranges over finite partitions of $X$ into measurable subsets.
When $f$ may be negative, its integral is
\[
        \int_X f\,\dd\mu \defeq \left(\int_X \max(0,f)\,\dd\mu\right) - \left(\int_X \max(0,-f)\,\dd\mu\right)
\]
when those two integrals exist.
When it is convenient to make the integrated variable explicit, we write $\int_{x \in U} f(x)\,\dd\mu$ for $\int_X (\lambda x.\,f(x) \cdot [x \in U])\,\dd\mu$, where $U \in \sigalg X$ is a measurable subset and $[\varphi]$ has the value $1$ if $\varphi$ holds
and $0$ otherwise.

\subsection{Standard Borel spaces}
\begin{proposition}[e.g.~\cite{kallenberg}, App.~A1]\label{prop:char-sbs}
  For a measurable space $(X,\sigalg X)$ the following are equivalent:
  \begin{itemize}
   \item $(X,\sigalg X)$ is a retract of $(\RR,\sigalg \RR)$, that is, there exist measurable \smash{$X\xrightarrow f \RR\xrightarrow g X$} such that $g \circ f=\id X$;
   \item $(X,\sigalg X)$ is either measurably isomorphic to $(\RR,\sigalg \RR)$ or countable and discrete;
   \item $X$ has a complete metric with a countable dense subset and $\sigalg X$ is the least $\sigma$-algebra containing all open sets.
  \end{itemize}
\end{proposition}
When $(X,\sigalg X)$ satisfies any of the above conditions,  we call it
\emph{standard Borel space}. These spaces play an important role in
probability theory because they enjoy properties that do not hold for
general measurable spaces, such as the existence of conditional
probability kernels~\cite{kallenberg,Preston-Borel08}
and de Finetti's theorem for exchangeable random processes~\cite{Dubins1979}. 

Besides $\RR$, another popular uncountable standard Borel space is $(0,1)$ with
the $\sigma$-algebra ${\{U \cap (0,1) ~|~ U \in \sigalg \RR\}}$. As
the above proposition indicates, these spaces are isomorphic by, for instance,
$\lambda r.\,\frac1{(1+e^{-r})} : \RR \to (0,1)$.

\subsection{Failure of cartesian closure}
\begin{proposition}[Aumann, \cite{aumann:functionspaces}]\label{prop:not-ccc}
  The category $\Meas$ is not cartesian closed: there is no space of functions $\RR\to\RR$. 
\end{proposition}
Specifically, the evaluation function 
\[\varepsilon\colon \Meas(\RR,\RR)\times \RR\to \RR
\qquad\text{with}\qquad
\varepsilon(f,r)=f(r)\]
is never measurable 
$(\Meas(\RR,\RR)\times \RR,\Sigma\otimes \sigalg \RR)\to (\RR,\sigalg \RR)$
regardless of the choice of $\sigma$-algebra $\Sigma$ 
on $\Meas(\RR,\RR)$.
Here, $\Sigma\otimes \sigalg \RR$ is the product $\sigma$-algebra,
generated by rectangles ${(U\times V)}$ for $U\in\Sigma$ and $V\in\sigalg \RR$. 



\section{Quasi-Borel spaces}\label{sec:quasiborel}
The typical situation in probability theory is that there is a fixed measurable space $(\Omega,\sigalg\Omega)$, called the \emph{sample space}, from which all randomness originates, and that observations are made in terms of random variables, which are pairs $(X,f)$ of a measurable space of observations $(X,\sigalg X)$ and a measurable function $f\colon \Omega\to X$.
From this perspective, the notion of measurable function is more important
than the notion of measurable space. In some ways, the
$\sigma$-algebra $\sigalg X$ is only used as an intermediary to restrain the class of measurable functions $\Omega\to X$.

We now use this idea as a basis for our new notion of space.
In doing so, we assume that our sample space $\Omega$ is the real numbers,
which makes probabilities behave well.
\begin{definition}\label{def:qbs}
A \emph{quasi-Borel space}
is a set $X$ together with a set $\qb X \subseteq [\RR\to X]$ satisfying:
\begin{itemize}
\item $\alpha \circ f \in \qb X$ if $\alpha \in \qb X$ and $f\colon \RR\to\RR$ is measurable;
\item $\alpha \in \qb X$ if $\alpha\colon \RR\to X$ is constant;
\item if ${\RR=\biguplus_{i\in \NN}S_i}$, with each set $S_i$ Borel,
  and ${\alpha_1,\alpha_2,\ldots\in \qb X}$, then
  ${\beta}$ is in ${\qb X}$, where ${\beta(r) = \alpha_i(r)}$ for $r\in S_i$.
\end{itemize}
\end{definition}
The name `quasi-Borel space' is motivated firstly by analogy to quasi-topological spaces (see \S\ref{sec:related}),
and secondly in recognition of the intimate connection to the standard Borel space $\RR$ (see also Prop.~\ref{prop:adjunction}(2)).

\begin{example}\normalfont
  For every measurable space $(X,\sigalg X)$, let $\sigtoqb {\sigalg X}$ be the set of measurable functions $\RR\to X$.
  Thus $\sigtoqb {\sigalg X}$ is the set of $X$-valued random variables.
  In particular: $\RR$ itself can be considered as a quasi-Borel space, with $\qb \RR$ the set of measurable functions $\RR\to\RR$;  the two-element discrete space $2$ can be considered as a quasi-Borel space, with $\qb 2$ the set of measurable functions $\RR\to 2$, which are exactly the characteristic functions of the Borel sets~(Def.~\ref{def:borel}).
\end{example}

Before we continue, we remark that the notion of quasi-Borel space
is invariant under replacing $\RR$ with a different uncountable standard Borel space.
\begin{proposition}
  For any measurable space $(\Omega,\sigalg\Omega)$, any measurable isomorphism $\iota \colon \RR \to \Omega$,
  any set $X$, and any set $N$ of functions $\Omega\to X$, the pair
  $(X,\{\alpha \circ \iota~|~\alpha\in N\})$
  is a quasi-Borel space if and only if:
  \begin{itemize}
   \item $\alpha \circ f \in N$ if $\alpha \in N$ and $f\colon \Omega\to\Omega$ is measurable;
   \item $\alpha \in N$ if $\alpha\colon \Omega\to X$ is constant;
\item if $\Omega=\biguplus_{i\in \NN}S_i$, with each set $S_i \in \sigalg \Omega$,
  and $\alpha_1,\alpha_2,\ldots\in N$, then
  $\beta$ is in $N$, where $\beta(r)=\alpha_i(r)$ if $r\in S_i$.
  \end{itemize}
\end{proposition}
By Prop.~\ref{prop:char-sbs}, the measurable spaces isomorphic to $\RR$ are the uncountable standard Borel spaces.
Note that the choice of isomorphism~$\iota$ is not important: it does not appear in the three conditions.

Probability theory typically considers a basic probability measure on the sample
space $\Omega$. Each random variable, that is each measurable function
$\Omega\to X$, then induces a probability measure on $X$ by pushing forward the basic measure. Quasi-Borel spaces take this idea as an axiomatic notion of probability measure.
\begin{definition}\label{def:probabilitymeasure}
  A \emph{probability measure} on a quasi-Borel space $(X,\qb X)$ is a pair $(\alpha,\mu)$ of $\alpha \in \qb X$ and a probability measure $\mu$ on $\RR$ (as in Def.~\ref{def:borelprob}).
\end{definition}

\subsection{Morphisms and integration}

\begin{definition}\label{def:morphism}
  A \emph{morphism} of quasi-Borel spaces $(X,\qb X)\to (Y,\qb Y)$ is a function $f \colon X \to Y$ such that $f \circ \alpha \in \qb Y$ if $\alpha\in \qb X$. Write $\QBS\big((X,\qb X),(Y,\qb Y)\big)$ for the set of morphisms from $(X,\qb X)$ to $(Y,\qb Y)$.
\end{definition}

In particular, elements of $M_X$ are precisely morphisms $(\RR,\qb\RR)\to (X,\qb X)$, so $\qb X = \QBS\big((\RR,\qb \RR),(X,\qb X)\big)$.

Morphisms compose as functions, and identity functions are morphisms,
so quasi-Borel spaces form a category $\QBS$.

\begin{example}\normalfont
  There are two canonical ways to equip a set $X$ with a quasi-Borel
  space structure. The first structure $\qb X^R$ consists of all
  functions $\RR \to X$. The second structure $\qb X^L$ consists of
  all functions $\beta \colon \RR \to X$ for which there exist: a countable
  subset $I \subseteq \NN$; a measurable $f \colon \RR \to \RR$; a
  partition $\RR = \biguplus_{i \in I}S_i$ with every $S_i$
  measurable; and a sequence $(x_i)_{i \in I}$ in $X$, such that
  $\beta(r) = x_i$ whenever $f(r) \in
  S_i$. These are the right and left adjoints, respectively, to the
  forgetful functor from $\QBS$ to $\Set$.
\end{example}

Def.~\ref{def:morphism} is independent of $\RR$: the sample space may be any uncountable standard Borel space.

\begin{proposition}
  Consider a measurable space $(\Omega,\sigalg\Omega)$
  with a measurable isomorphism $\iota \colon \RR \to \Omega$.
  For $i \in {1,2}$, let $X_i$ be a set and $N_i$ a set of functions $\Omega\to X_i$ such that
  \[
    M_i = \big(X_i,\{\alpha \circ \iota~|~\alpha\in N_i\} \big)
  \]
  are quasi-Borel spaces. A function $g \colon X_1 \to X_2$ is a morphism $(X_1,M_1) \to (X_2,M_2)$ if and only if $g \circ \alpha \in N_2$ for $\alpha \in N_1$.
\end{proposition}

Morphisms between quasi-Borel spaces are analogous to measurable functions
between measurable spaces. The crucial properties of measurable functions
are that they work well with (probability) measures: we can push-forward these measures, and integrate over them.
Morphisms of quasi-Borel spaces also support these constructions.
\begin{itemize}
  \item\emph{Pushing forward:} if $f\colon X\to Y$ is a morphism and $(\alpha,\mu)$ is a probability measure on $X$ then $f\circ \alpha$ is by definition in $\qb Y$ and so $(f\circ\alpha,\mu)$ is a probability measure on $Y$.
  \item\emph{Integrating:} If $f\colon X\to \RR$ is a morphism of quasi-Borel spaces and $(\alpha,\mu)$ is a probability measure on $X$, the integral of $f$ with respect to $(\alpha,\mu)$ is
  \begin{equation}\label{eqn:def-qbs-integration}
    \int f\, \dd(\alpha,\mu)\defeq \int_\RR (f\circ \alpha)\,\dd\mu\text.
  \end{equation}
  So integration formally reduces to integration on $\RR$.
\end{itemize}

\subsection{Relationship to measurable spaces}

If we regard a subset $S\subseteq X$ as its characteristic function $\chi_S\colon X\to 2$, then we can regard a $\sigma$-algebra on a set $X$ as a set of characteristic functions $F_X\subseteq [X\to 2]$ satisfying certain conditions. Thus a measurable space (Def.~\ref{def:meas-space}) could equivalently be described
as a pair $(X,F_X)$ of a set $X$ and a collection $F_X\subseteq [X\to 2]$ of characteristic functions.
Moreover, from this perspective, a measurable function $f\colon (X,F_X)\to (Y,F_Y)$ is
simply a function $f\colon X\to Y$ such that $\chi\circ f \in F_X$ if $\chi\in F_Y$.
Thus quasi-Borel spaces shift the emphasis from characteristic functions $X\to 2$ to random variables $\RR\to X$.

\subsubsection{Quasi-Borel spaces as structured measurable spaces}
A subset $S\subseteq X$ is in the $\sigma$-algebra $\sigalg X$
of a measurable space $(X,\sigalg X)$ if and only if its
characteristic function $X\to 2$ is measurable.
With this in mind, we \emph{define} a measurable subset of a
quasi-Borel space $(X,\qb X)$ to be a subset $S\subseteq X$
such that the characteristic function $X\to 2$ is a morphism of quasi-Borel spaces.
\begin{proposition}
  The collection of all measurable subsets of a quasi-Borel space $(X,\qb X)$ is characterized as
  \begin{equation}\label{eqn:qb-sigma}
    \qbtosig {\qb X}\defeq \{U~|~\forall \alpha\in\qb X.\,\inv\alpha(U)\in\sigalg \RR\}
  \end{equation}
  and forms a $\sigma$-algebra.
\end{proposition}
Thus we can understand a quasi-Borel space as a measurable space $(X,\sigalg X)$ equipped
with a class of measurable functions $\qb X\subseteq [\RR\to X]$ determining the $\sigma$-algebra by $\sigalg X=\qbtosig{\qb X}$
as in~\eqref{eqn:qb-sigma}.

Moreover, every morphism $(X,\qb X)\to (Y,\qb Y)$ is also a measurable
function $(X,\qbtosig {\qb X})\to (Y,\qbtosig{\qb Y})$ (but the converse does not hold
in general).

A probability measure $(\alpha,\mu)$ on a quasi-Borel space $(X,\qb X)$ induces a
probability measure $\alpha_*\mu$ on the underlying measurable space.
Integration as in~\eqref{eqn:def-qbs-integration} matches the standard
definition for measurable spaces.

\subsubsection{An adjunction embedding standard Borel spaces}
Under some circumstances morphisms of quasi-Borel spaces coincide with measurable functions.
\begin{proposition}\label{prop:adjunction}
  Let $(Y,\sigalg Y)$ be a measurable space.
  \begin{enumerate}
  \item If $(X,\qb X)$ is a quasi-Borel space, a function $X\to Y$ is a measurable function
  $(X,\qbtosig{\qb X})\to (Y,\sigalg Y)$ if and only if it is a morphism $(X,\qb X)\to (Y,\sigtoqb {\sigalg Y})$.
  \item If $(X,\sigalg X)$ is a standard Borel space, a function $X\to Y$ is a morphism $(X,\sigtoqb{\sigalg X})\to (Y,\sigtoqb {\sigalg Y})$ if and only if it is a measurable function $(X,\sigalg X)\to (Y,\sigalg Y)$.
\end{enumerate}
\end{proposition}
\newcommand{\meastoqbs}{R}
\newcommand{\qbstomeas}{L}
Proposition~\ref{prop:adjunction}(1) means there is an adjunction
\[\xymatrix{\Meas\ar@/_/[rr]_\meastoqbs&&\ar@/_/[ll]_\qbstomeas^\bot \QBS}\]
where
$\qbstomeas(X,\qb X)= (X,\qbtosig{\qb X})$
and
$\meastoqbs(X,\sigalg X)= (X,\sigtoqb{\sigalg X})$.
Proposition~\ref{prop:adjunction}(2) means that the functor $\meastoqbs$ is full and faithful when restricted to standard Borel spaces. Equivalently,
$L(R(X,\sigalg X))=(X,\sigalg X)$, that is $\sigalg X = \qbtosig {\sigtoqb {\sigalg X}}$
for standard Borel spaces $(X,\sigalg X)$.


\section{Products, coproducts and function spaces}\label{sec:structure}
Quasi-Borel spaces support products, coproducts, and function spaces.
These basic constructions form the basis for interpreting simple type theory
in quasi-Borel spaces. 
\begin{proposition}[Products]\label{prop:qbsproduct}
  If $(X_i,\qb {X_i})_{i\in I}$ is a family of quasi-Borel spaces indexed by a set $I$, then $(\prod_iX_i,\qb {\Pi_i X_i})$ is a quasi-Borel space, where $\prod_iX_i$ is the set product, and 
  \[
    \qb {\Pi_iX_i}\defeq \textstyle{\Big\{f\colon \RR\to \prod_iX_i~|~\forall i.\, (\pi_i\circ f)\in\qb{X_i}\Big\}\text.}
  \]
  The projections $\prod_i X_i\to X_i$ are morphisms, and provide the structure of a categorical product in $\QBS$.
\end{proposition}

\begin{proposition}[Coproducts]\label{prop:qbscoproduct}
  If $(X_i,\qb{X_i})_{i\in I}$ is a family of quasi-Borel spaces indexed by a countable set $I$, then $(\coprod_iX_i,\qb{\amalg_i X_i})$ is a quasi-Borel space, where $\coprod_iX_i$ is the disjoint union of sets, 
  \begin{align*}
    \qb{\amalg_i X_i}\defeq \{\lambda r.\,(f(r),\alpha_{f(r)}(r))
    \mid\; &f\colon \RR\to I\ \mbox{is measurable},\\
    & (\alpha_{i}\in \qb {X_{i}})_{i\in \mathsf{image}(f)}\}\text,
  \end{align*}
  and $I$ carries the discrete $\sigma$-algebra.
  This space has the universal property of a coproduct in the category $\QBS$.
\end{proposition}
\begin{proof}[Proof notes]
  The third condition of quasi-Borel spaces is needed here. 
  It is a crucial step in showing that for an $I$-indexed family of morphisms $(f_i\colon X_i\to Z)_{i\in I}$, the copairing $[f_i]_{i\in I}\colon\coprod_{i\in I}X_i\to Z$ is again a morphism.
\end{proof}

\begin{proposition}[Function spaces]\label{prop:functionspace}
  If $(X,\qb X)$ and $(Y,\qb Y)$ are quasi-Borel spaces, so is $(Y^X,\qb {Y^X})$,
  where ${Y^X\defeq \QBS(X,Y)}$ is the set of morphisms $X\to Y$, and 
  \[
    \qb{Y^X}\defeq \{\alpha\colon\RR\to Y^X \mid
	\mathsf{uncurry}(\alpha)\in\QBS(\RR\times X,Y)\}\text.
  \] 
  The evaluation function $Y^X\times X\to Y$ is a morphism and has the universal property of the function space.
  Thus $\QBS$ is a cartesian closed category.
\end{proposition}
\begin{proof}[Proof notes]
  The only difficult part is showing that $(Y^X,\qb{Y^X})$ satisfies the third condition of quasi-Borel spaces. 
  Prop.~\ref{prop:qbscoproduct} is useful here.
\end{proof}
\subsection{Relationship with standard Borel spaces}
Recall that standard Borel spaces can be thought of as a full subcategory 
of the quasi-Borel spaces, that is,
the functor $\meastoqbs\colon \Meas\to \QBS$ is full and faithful (Prop.~\ref{prop:adjunction}(2))
when restricted to the standard Borel spaces.
This full subcategory has the same countable products, coproducts and function spaces 
(whenever they exist). 
We may thus regard quasi-Borel spaces as a conservative extension of standard Borel spaces that supports simple type theory.

\begin{proposition} 
  The functor $\meastoqbs(X,\sigalg X)=(X,\sigtoqb {\sigalg X})$:
  \begin{enumerate} 
  \item preserves products of standard Borel spaces:
  $\meastoqbs(\prod_iX_i)=\prod_i\meastoqbs(X_i)$, where $(X_i,\sigalg{X_i})_{i\in I}$ is a countable family of standard Borel spaces;
  \item preserves spaces of functions between standard Borel spaces whenever they exist:
  if $(Y,\sigalg Y)$ is countable and discrete, 
  and $(X,\sigalg X)$ is standard Borel, then $\meastoqbs(X^Y)=\meastoqbs(X)^{\meastoqbs(Y)}$;
  \item preserves countable coproducts of standard Borel spaces:
  $\meastoqbs(\coprod_iX_i)=\coprod_i\meastoqbs(X_i)$, where $(X_i,\sigalg{X_i})_{i\in I}$ is a countable family of standard Borel spaces. 
\end{enumerate}
\end{proposition}
Consequently, a standard programming language semantics in standard Borel spaces 
can be conservatively embedded in quasi-Borel spaces, allowing higher-order functions while preserving all the type theoretic structure. 

We note, however, that in light of Prop.~\ref{prop:not-ccc},
the quasi-Borel space $\RR^\RR$ does not come from a standard Borel space. 
Moreover, the left adjoint $L\colon \QBS\to \Meas$ does not preserve products
in general.
For quasi-Borel spaces $(X,\qb X)$ and $(Y,\qb Y)$, we always have 
$\qbtosig{\qb X}\otimes \qbtosig{\qb Y}\subseteq \qbtosig {\qb {X\times Y}}$, 
but not always $\supseteq$. Indeed, $\qbtosig{\qb {\RR^\RR}}\otimes \sigalg \RR
\neq \qbtosig {\qb{(\RR^\RR\times \RR)}}$, by Prop.~\ref{prop:not-ccc}.

\section{A monad of probability measures}\label{sec:giry}
In this section we will show that the probability measures on a quasi-Borel space form a quasi-Borel space again. This gives a commutative monad that generalizes the Giry monad for measurable spaces~\cite{giry:monad}.
\subsection{Monads}
\newcommand{\CCat}{\mathcal C}
We use the Kleisli triple formulation of monads (see e.g.~\cite{moggi-monads}). Recall that a monad on a category $\CCat$ comprises 
\begin{itemize}
\item for any object $X$, an object $T(X)$;
\item for any object $X$, a morphism $\munit\colon X\to T(X)$;
\item for any objects $X,Y$, a function
 \[\bindname\colon \CCat(X,T(Y))\to \CCat(T(X),T(Y))\text.\]
  We write $(t\bindsymbol f)$ for $\bindname(f)(t)$. 
\end{itemize}
This is subject to the conditions $( t\bindsymbol \munit)=t$, $( {\munit(x)}\bindsymbol f)=f(x)$, and 
$\bind t {(\lambda x.\,(\bind {f(x)} g))}=
\bind{(\bind t f)}g$. 

The intuition is that $T(X)$ is an object of computations returning~$X$, 
that $\munit$ is the computation that returns immediately, and that 
$t\bindsymbol f$ sequences computations,
first running computation $t$ and then calling $f$ with the result.

When $\CCat$ is cartesian closed, 
a monad is \emph{strong} if $\bindname$ internalizes to an operation
$\bindname\colon (T(Y))^X\to (T(Y))^{T(X)}$,
and then the conditions are understood as expressions in a cartesian closed category.

\subsection{Kernels and the Giry monad}
We recall the notion of probability kernel, which is a measurable family of probability measures. 
\begin{definition}
  Let $(X,\sigalg X)$ and $(Y,\sigalg Y)$ be measurable spaces.
  A \emph{probability kernel} from $X$ to $Y$ is a function $k:X\times \sigalg Y\to [0,1]$ such that $k(x,-)$ is a probability measure for all $x\in X$ (Def.~\ref{def:meas-space}), and $k(-,U)$ is a measurable function for all $U\in\sigalg Y$ (Def.~\ref{def:measurablefunction}). 
\end{definition}
We can classify probability kernels as follows.
Let $\Giry(X)$ be the set of probability measures on $(X,\sigalg X)$. 
We can equip this set with the $\sigma$-algebra generated by 
${\{\mu\in\Giry(X)~|~\mu(U)<r\}}$, for $U\in\sigalg X$ and $r\in[0,1]$,
to form a measurable space $(\Giry(X),\sigalg{\Giry(X)})$.
A measurable function $X\to \Giry(Y)$ amounts to a probability kernel from $X$ to $Y$. 

The construction $\Giry$ has the structure of a monad, as first discussed by
Giry~\cite{giry:monad}.
A computational intuition is that $\Giry(X)$ is a space of probabilistic computations
over $X$, and this provides a semantic foundation for a first-order probabilistic programming language (see e.g.~\cite{statonyangheunenkammarwood:higherorder}). 
The unit $\eta\colon X\to \Giry(X)$ lets $\eta(x)$ be the Dirac measure on $x$, 
with $\eta(x)(U)= 1$ if $x\in U$, and $\eta(x)(U)= 0$ if $x\not\in U$.
If $\mu\in \Giry(X)$ and $k$ is a measurable function $X\to \Giry(Y)$, then $(\mu\bindGsymbol k)$ is the measure in $\Giry(Y)$ with 
$(\mu\bindGsymbol k)(U)= \int_{x \in X} k(x)(U) \,\dd\mu$. 

\subsection{Equivalent measures on quasi-Borel spaces}
Recall (Def.~\ref{def:probabilitymeasure}) that a probability measure $(\alpha,\mu)$ on a quasi-Borel space $(X,\qb X)$ is a pair $(\alpha,\mu)$ of a function 
$\alpha\in\qb X$ and a probability measure $\mu$ on $\RR$. 
Random variables are often equated when they describe the same distribution. 
Every probability measure $(\alpha,\mu)$ determines 
a push-forward measure $\alpha_* \mu$ on the corresponding 
measurable space $(X,\qbtosig{\qb X})$, that assigns to $U \subseteq X$ the real number $\mu(\alpha^{-1}(U))$.
We will identify two probability measures when they define the same push-forward measure, and write $\sim$ for this equivalence relation. 

This is a reasonable notion of equality even if we put aside the notion of 
measurable space, because two probability measures have the same push-forward measure precisely when they have the same integration operator: $(\alpha,\mu) \sim (\alpha',\mu')$ if and only if $\int f\, \dd(\alpha,\mu) = \int f\,\dd(\alpha',\mu')$ for all morphisms ${f \colon (X,\qb X) \to \RR}$.
Nevertheless, other notions of equivalence could be used.

\subsection{A probability monad}
We now explain how to build a monad of probability measures on the 
category of quasi-Borel spaces, modulo this notion of 
equivalence. 
This monad $\Pmonad$ will inherit properties 
from the Giry monad.
Technically, the functor
$L\colon \QBS\to \Meas$ (Prop.~\ref{prop:adjunction}) is a `monad opfunctor' 
taking $\Pmonad$ to the 
Giry monad $\Giry$, which means that it extends to a 
functor from the Kleisli category of~$\Pmonad$ to the Kleisli category of~$\Giry$~\cite{street-monads}.
\paragraph{On objects}
For a quasi-Borel space $(X,\qb X)$, let 
\begin{align*}
    P(X) &= \{ (\alpha,\mu) \text{ probability measure on } (X,\qb X)\}\slash\sim \text{,}\\
    \qb{P(X)} & = \{ \beta \colon \RR \to P(X) \mid \exists \alpha \in \qb X.\,\exists g \in \Meas(\RR, G(\RR)).\,
    \\
    & \phantom{= \{ \beta \colon \RR \to P(X) \mid\;\;} \forall r \in \RR.\, 
    \beta(r) = [\alpha,g(r)] \}  \text{,}
\end{align*}
where $[\alpha,\mu]$ denotes the equivalence class.
Note that \begin{equation}P(X)\cong \{\alpha_*\mu\in \Giry(X,\qbtosig{\qb X})~|~\alpha\in\qb X,\,\mu\in\Giry(\RR)\}\label{eqn:opfunctor-map}\end{equation}
as sets, and $l_X([\alpha,\mu])= \alpha_*\mu$ defines a measurable injection 
$l_X\colon L(P(X))\rightarrowtail \Giry (X,\qbtosig{\qb X})$.

\paragraph{Monad unit (return)}
Recall that the constant functions $(\lambda r.x)$ are all in $\qb X$.
For any probability measure $\mu$ on $\RR$, the push-forward measure
$(\lambda r.x)_*\,\mu$ on $(X,\qbtosig{\qb X})$ 
is the Dirac measure on $x$, with $((\lambda r.x)_*\,\mu)(U)=1$ if $x\in U$
and $0$ otherwise. 
Thus $(\lambda r.x,\mu) \sim (\lambda r.x,\mu')$ for all measures  $\mu,\mu'$ on $\RR$.
The unit of $\Pmonad$ at $(X,\qb X)$ is the morphism 
$\munit \colon X\to \Pmonad(X)$ given by
\begin{equation}\label{eq:giry:unit}
  \munit_{(X,\qb X)}(x) = [ \lambda r.x,\, \mu ]
\end{equation}
for an arbitrary probability measure $\mu$ on $\RR$. 

\paragraph{Bind}
To define ${\bindname\colon \Pmonad(Y)^X\to (\Pmonad(Y))^{\Pmonad(X)}}$,
suppose $f\colon X\to \Pmonad(Y)$ 
is a morphism and $[\alpha,\mu]$ in $\Pmonad(X)$.
Since $f$ is a morphism, there is a measurable $g\colon \RR\to\Giry(\RR)$ 
and a function $\beta\in \qb Y$ such that 
$(f\circ \alpha)(r) = [\beta,g(r)]$. 
Set $(\bind{[\alpha,\mu]} f) = [\beta,\bindG \mu g]$,
where $\bindG\mu g$ is the bind of the Giry monad.
This matches the bind of the Giry monad, since 
$(\bindG{(\alpha_*\mu)}{(l_Y \circ f)})=\beta_*(\bindG \mu g)$. 
\begin{theorem}
\label{thm:Pmonad}
The data $(P,\eta,\bindname)$ above defines a strong monad on the 
category $\QBS$ of quasi-Borel spaces.
\end{theorem}
\begin{proof}[Proof notes]
  The monad laws can be reduced to the laws for the monad $\Giry$ on $\Meas$~\cite{giry:monad}. The monad on $\QBS$ is strong because
  $\bindname\colon \Pmonad(Y)^X\to (\Pmonad(Y))^{\Pmonad(X)}$ is a 
  morphism, which is shown by expanding the definitions. 
\end{proof}

\begin{proposition}\label{prop:giryproperties}
  The monad~$\Pmonad$ satisfies these properties:
\begin{enumerate}
\item For $f\colon (X,\qb X)\to (Y,\qb Y)$, the functorial action 
$\Pmonad(f)\colon P(X)\to P(Y)$ is $[\alpha,\mu] \mapsto [f\circ \alpha,\mu]$.
\item It is a commutative monad, i.e.\ the order of 
sequencing doesn't matter:
if $p\in \Pmonad(X)$, $q\in\Pmonad(Y)$, and $f$ is a morphism $X\times Y\to \Pmonad(Z)$, then $\bind p {\lambda x.\,\bind q{\lambda y.\,f(x,y)}}$ equals $\bind q {\lambda y.\,\bind p{\lambda x.\,f(x,y)}}$.
\item The faithful functor ${L\colon \QBS\to \Meas}$ with ${L(X,\qb X)=(X,\qbtosig{\qb X})}$ extends to a faithful functor ${\mathsf{Kleisli}(\Pmonad)\to \mathsf{Kleisli}(\Giry)}$, i.e.\ $(L,l)$ is a monad opfunctor~\cite{street-monads}.
\item When $(X,\sigalg X)$ is a standard Borel space, the map
  $l_X$ of Eq.~\eqref{eqn:opfunctor-map} is a measurable isomorphism.
\label{prop:monad-morphism}
\end{enumerate}
\end{proposition}


\section{Example: Bayesian regression}
\label{sec:example}
We are now in a position to explain the semantics of the Anglican program in Figure~\ref{fig:linearregression}.
The program can be split into three parts: a prior, a likelihood, and a posterior. 
Recall that Bayes' law says that the posterior is proportionate to the 
product of the prior and the likelihood.

\paragraph{Prior}
Lines 2--4 define a prior measure on $\RR^\RR$:
\newcommand*{\mlstinline}[1]{\text{\lstinline|#1|}}
\begin{align*}
\mathit{prior}\defeq 
\mlstinline{(let [}&\mlstinline{s (sample (normal 0.0 3.0)) }\\[-4pt]&\mlstinline{b (sample (normal 0.0 3.0))]}\\[-4pt]&\hspace{-0.5cm}\mlstinline{(fn [x] (+ (* s x) b)))}
\end{align*}

To describe this semantically, observe the following. 
\begin{proposition}
        \label{prop:change-randomsource}
Let $(\Omega,\sigalg \Omega)$ be a standard Borel space,
and $(X,\qb X)$ a quasi-Borel space. 
Let $\alpha\colon R(\Omega,\sigalg \Omega)\to X$ be a morphism and $\mu$ a probability measure on $(\Omega,\sigalg \Omega)$. 
Any section-retraction pair $(\Omega\xrightarrow\varsigma \RR \xrightarrow \rho\Omega)=\id\Omega$ has a probability measure $[\alpha\circ \rho,\varsigma_* \mu]\in \Pmonad(X)$, that is independent of the choice of $\varsigma$ and $\rho$. 
\label{prop:different-prob-space}
\end{proposition}
Write $[\alpha,\mu]$ for the probability measure in this case. 

Now, the program fragment $\mathit{prior}$ describes the distribution $[\alpha,\nu\otimes \nu]$ 
in $\Pmonad(\RR^\RR)$ 
where $\nu$ is the normal distribution on $\RR$ with mean $0$ and standard deviation $3$,
and where $\alpha\colon \RR\times \RR\to \RR^\RR$ is given by 
$\alpha(s,b)\defeq \lambda r.\,s \cdot r+b$. 
Informally, 
\begin{equation}
\label{eqn:linear-reg-prior}
\denot{\mathit{prior}}=[\alpha,\nu\otimes\nu]\in \Pmonad(\RR^\RR)\text.
\end{equation}
Figure~\ref{fig:prior} illustrates this measure $[\alpha, \nu \otimes \nu]$.
This denotational semantics can be made compositional, by using 
the commutative monad structure of $\Pmonad$ and the cartesian closed structure of the category $\QBS$ (following e.g.~\cite{moggi-monads,statonyangheunenkammarwood:higherorder}),
but in this paper we focus on this example rather than spelling out the general case once again.

\begin{figure}
\includegraphics[width=\linewidth]{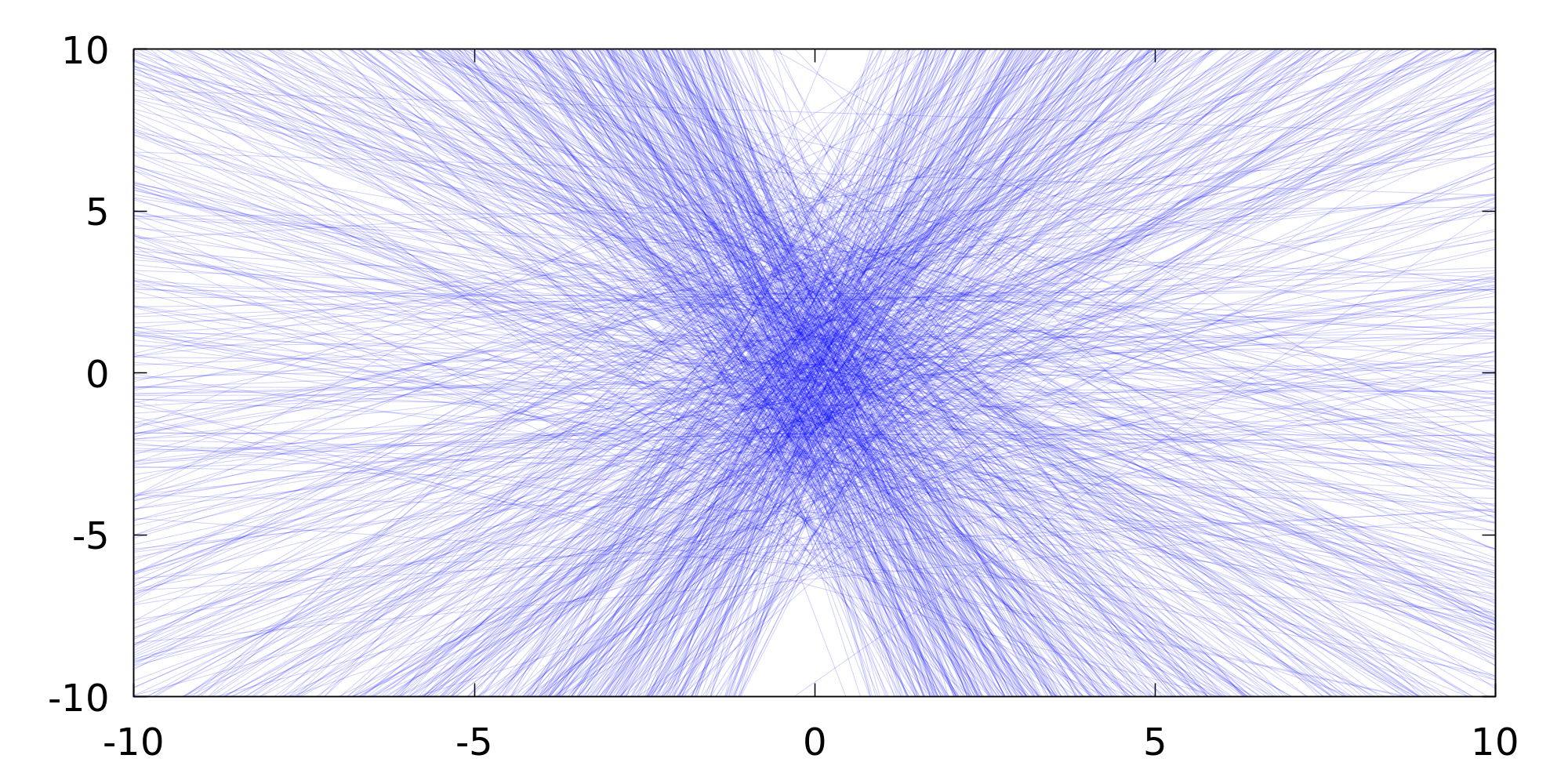}
\caption{Illustration of 1000 sampled functions from the prior on $\RR^\RR$ for Bayesian linear regression (\ref{eqn:linear-reg-prior}).}
\label{fig:prior}
\end{figure}

\newcommand{\nordens}{d}

\paragraph{Likelihood}Lines 5--9 define the likelihood of the observations:
\begin{align*}
\mathit{obs}\ \ \defeq &\phantom{\mbox{} \dots \mbox{} }    \mlstinline{(observe (normal (f 1.0) 0.5) 2.5)}\\&\dots
    \mlstinline{
    (observe (normal (f 5.0) 0.5) 8.0)}
\end{align*}
This program fragment has a free variable~$f$ of type~$\RR^\RR$. 
Let us focus on line~5 for a moment:
\[\mathit{obs}_1\defeq \mlstinline{(observe (normal (f 1.0) 0.5) 2.5)}\]
Given a function $f\colon \RR\to \RR$,
the likelihood of drawing $2.5$ from a normal distribution 
with mean $f(1.0)$ and standard deviation $0.5$ is 
\[
\denot{f\colon\RR^\RR\vdash \mathit{obs}_1}=\nordens(f(1.0),2.5)\text,
\]
where $\nordens\colon \RR^2\to[0,\infty)$ is the density of the normal distribution function with standard deviation $0.5$:
\[
\nordens(\mu,x)
\ =\ 
\sqrt{\tfrac 2\pi}\,e^{-2(x-\mu)^2}\text.
\]
Notice that we use a normal distribution 
to allow for some noise in the measurement. Informally, we are not recording an observation that $f(1.0)$ is \emph{exactly} $2.5$,
since this would make regression impossible; rather, $f(1.0)$ is roughly $2.5$.

Overall, lines 5--9 describe a likelihood weight which is 
the product of the likelihoods of the five data points, given 
$f\colon \RR^\RR$. 
\begin{align*}
\denot{f\colon \RR^\RR\vdash\mathit{obs}}
= \nordens(f(1),2.5)
&
{}
\cdot \nordens(f(2),3.8)
\cdot \nordens(f(3),4.5)
\\
&
{}
\cdot \nordens(f(4),6.2)
\cdot \nordens(f(5),8.0)\text.
\end{align*}

\paragraph{Posterior}
We follow the recipe for a semantic posterior given in~\cite{statonyangheunenkammarwood:higherorder}.
Putting the prior and likelihood together gives a probability measure 
in~$\Pmonad(\RR^\RR\times [0,\infty))$
which is found by pushing forward the measure 
$\denot{\mathit{prior}}\in\Pmonad(\RR)$ along the function 
\shortversion{$(\id\,,\,\denot{\mathit{obs}})\colon \RR^\RR\to \RR^\RR\times [0,\infty)$.}
\longversion{\[(\id\,,\,\denot{\mathit{obs}})\colon \RR^\RR\to \RR^\RR\times [0,\infty)\text.\]}
This push-forward measure
\[\Pmonad(\id\,,\,\denot{\mathit{obs}})\,(\denot{\mathit{prior}})\in \Pmonad(\RR^\RR\times [0,\infty))\]
is a measure over pairs $(f,w)$ of functions together with their likelihood weight.
We now find the posterior by multiplying the prior and the likelihood, and dividing by a normalizing constant.
To do this we define a morphism
\newcommand{\normalize}{\mathit{norm}}%
\begin{align*}
&
\normalize\colon \Pmonad(X\times [0,\infty))\to \Pmonad(X)\uplus\{\mathsf{error}\}
\\
&\normalize([(\alpha,\beta),\nu])\defeq 
\begin{cases}
[\alpha,\nu_\beta/(\nu_\beta(\RR))]&\text{if }{0\neq \nu_\beta(\RR)\neq \infty}\\
\mathsf{error}&\text{otherwise}\end{cases}
\end{align*}
where $\nu_\beta\colon \sigalg\RR\to [0,\infty] 
\defeq \lambda U.\,\int_{r \in U}(\beta(r))\, \dd \nu$.
The idea is that if $\beta\colon\RR\to [0,\infty)$ and $\nu$ is a probability measure on $\RR$ then 
$\nu_\beta$ is always a posterior measure on $\RR$, but it is typically not normalized, 
i.e.~$\nu_\beta(\RR)\neq 1$. We normalize it by dividing by the normalizing constant, 
as long as this division is well-defined.

Now, the semantics of the entire program in Figure~\ref{fig:linearregression}
is $\normalize(\Pmonad(\id\,,\,\denot{\mathit{obs}})\,(\denot{\mathit{prior}}))$, which is a measure in $\Pmonad(\RR^\RR)$. Calculating this posterior using Anglican's
inference algorithm \texttt{lmh} gives the plot in the lower half of 
Figure~\ref{fig:linearregression}.

\paragraph{Defunctionalized regression and non-linear regression}
Of course, one can do regression without explicitly considering distributions over the space of all measurable functions, 
by instead directly calculating posterior distributions for the slope~$s$ and the intercept~$b$.
For example, one could defunctionalize the program in Fig.~\ref{fig:linearregression} in the style of Reynolds~\cite{defunctionalization}.
But defunctionalization is a whole-program transformation. 
By structuring the semantics using quasi-Borel spaces, we are able to work compositionally, without mentioning~$s$ and~$b$ explicitly on lines 5--10. 
The internal posterior calculations actually happen at the level of standard Borel spaces, and so a defunctionalized version would be in some sense equivalent,
but from the programming perspective it helps to abstract away from this.
The regression program in Fig.~\ref{fig:linearregression} is quickly adapted to fit other kinds of functions, e.g.\ polynomials, or even programs from a small domain-specific language, simply by changing the prior in Lines~2--4. 


\section{Random functions}\label{sec:functions}
We discuss random variables and random functions, starting from the traditional setting.
Let $(\Omega,\sigalg\Omega)$ be a measurable space with a probability measure.
A \emph{random variable} is a measurable function
$(\Omega,\sigalg\Omega)\to (X,\sigalg X)$.
A \emph{random function} between measurable spaces
$(X,\sigalg X)$ and $(Y,\sigalg Y)$ is a measurable function
$(\Omega \times X,\sigalg \Omega \otimes \sigalg X)\to (Y,\sigalg Y)$.

We can push forward a probability measure on $\Omega$ along a random variable
$(\Omega,\sigalg\Omega)\to (X,\sigalg X)$ to get a probability measure on
$X$, but in the traditional setting we cannot push forward a measure along a
random function. Measurable spaces are not cartesian closed (Prop.~\ref{prop:not-ccc}),
and so we cannot form a measurable space $Y^X$ and we cannot curry a random function in general.

Now, if we revisit these definitions in the setting of quasi-Borel spaces,
we \emph{do} have function spaces, and so we can push forward along random functions.
In fact, this is somewhat tautologous because a probability measure (Def.~\ref{def:probabilitymeasure}) on a function space
is essentially the same thing as a random function: a probability measure on a function
space $(Y,\sigalg Y)^{(X,\sigalg X)}$ is defined to be a pair $(f,\mu)$ of a
probability measure $\mu$ on $\RR$, our sample space, and a morphism
$f \colon \RR\to Y^X$; but to give a morphism $\RR\to Y^X$ is to give a
morphism $\RR \times X \to Y$ (Prop.~\ref{prop:functionspace})
as in the traditional definition of random function.

We have already encountered an example of a random function in Section~\ref{sec:example}: the prior for linear regression is a random function from $\RR$ to $\RR$
over the measurable space $(\RR\times \RR, \sigalg \RR \otimes \sigalg \RR)$ with the measure $\nu\otimes \nu$.
Random functions abound throughout probability theory and stochastic processes.
The following section explores their use in the so-called randomization lemma, which is
used throughout probability theory. By moving to quasi-Borel spaces,
we can state this lemma succinctly (Theorem~\ref{theorem:random-quotient}).

\subsection{Randomization}
An elementary but useful trick in probability theory is
that every probability distribution on $\RR$
arises as a push-forward of the uniform distribution on $[0,1]$.
Even more useful is that this can be done in a parameterized way.
\begin{proposition}[\cite{kallenberg}, Lem.~3.22]\label{prop:randomization}
Let $(X,\sigalg X)$ be a measurable space.
For any kernel ${k\colon X\times \sigalg \RR\to [0,1]}$
there is a measurable function ${f\colon\RR\times X\to \RR}$ such that
${k(x,U)=\upsilon\{r~|~f(r,x)\in U\}}$, where
$\upsilon$ is the uniform distribution on $[0,1]$.
\end{proposition}
For quasi-Borel spaces we can phrase this more succinctly: it is a result
about a quotient of the space of random functions.
We first define quotient spaces.
\begin{proposition}
Let $(X,\qb X)$ be a quasi-Borel space, let $Y$ be a set,
and let $q\colon X\to Y$ be a surjection.
Then $(Y,\qb Y)$ is a quasi-Borel space with $\qb Y=\{q\circ \alpha~|~\alpha\in \qb X\}$.
\end{proposition}
We call such a space a \emph{quotient} space.
\begin{theorem}\label{theorem:random-quotient}
Let $(X,\qb X)$ be a quasi-Borel space.
The space $(\Pmonad(\RR))^X$ of kernels is a quotient of the
space $\Pmonad(\RR^X)$ of random functions.
\end{theorem}
Before proving this theorem, we use Prop.~\ref{prop:randomization}
to give an alternative characterization of our probability monad.
\begin{lemma}\label{lemma:alt-qb-pmonad}
Let $(X,\qb X)$ be a quasi-Borel space.
The function $q\colon X^\RR\to \Pmonad(X)$ given by $q(\alpha)\defeq [\alpha,\upsilon]$
is a surjection, with corresponding quotient space
$(\Pmonad(X),\qb{\Pmonad(X)})$:
\begin{align}
\label{eqn:altMPmonad}
M_{\Pmonad(X)}&=\{\lambda r\in\RR.\,[\gamma(r),\upsilon]~|~\gamma\in \qb{X^\RR}\}\text,
\end{align}
where $\upsilon$ is the uniform distribution on $[0,1]$.
\end{lemma}
\begin{proof}[Proof notes]
The direction $(\subseteq)$ follows immediately from Prop.~\ref{prop:randomization}.
For the direction $(\supseteq)$ 
we must consider $\gamma\in \qb{X^\RR}$ and show that
$(\lambda r\in \RR.\,[\gamma(r),\upsilon])$ is in $\qb{\Pmonad(X)}$.
This follows by considering the kernel $k\colon \RR\to \Giry(\RR\times \RR)$ with
$k(r)=\upsilon\otimes\delta_r$,
so that $[\gamma(r),\upsilon]=[\uncurry(\gamma),k(r)]$. Here we are using Prop.~\ref{prop:change-randomsource}.
\end{proof}
\shortversion{
\begin{proof}[Proof of Theorem~\ref{theorem:random-quotient}]
Consider the  evident morphism
$q\colon \Pmonad(\RR^X)\to (\Pmonad(\RR))^X$
that comes from the monadic strength.
That is, $(q([\alpha,\mu]))(x)=[\lambda r.\,\alpha(r)(x),\mu]$.
We show that $q$ is a quotient morphism.

We first show that $q$ is surjective.
To give a morphism $k\colon (X,\qb X)\to \Pmonad(\RR)$
is to give a measurable function $(X,\qbtosig{\qb X})\to\Giry(\RR)$,
since
$(\Pmonad(\RR),\qb{\Pmonad(\RR)})\cong (\Giry(\RR),\sigtoqb{\sigalg{\Giry(\RR)}})$ (Prop.~\ref{prop:giryproperties}(4))
 and by using the adjunction between measurable spaces
and quasi-Borel spaces (Prop.~\ref{prop:adjunction}(1)).
Directly, we understand a morphism $k\colon  (X,\qb X)\to \Pmonad(\RR)$
as the kernel $k^\sharp\colon X\times \sigalg \RR\to [0,1]$ with
$k^\sharp(x,U)\defeq \mu_x(\inv{\alpha}_x(U))$ whenever
$k(x)=[\alpha_x,\mu_x]$. The definition of $k^\sharp$ does not depend on the choice
of $\alpha_x,\mu_x$.

Now we can
use the randomization lemma (Prop.~\ref{prop:randomization})
to find a measurable function $f_{k^\sharp}\colon \RR\times X\to \RR$ such that
$k^\sharp(x,U)=\upsilon\{r \mid f_{k^\sharp}(r,x)\in U\}$.
In general, if a function ${Y\times X\to Z}$ is jointly measurable
then it is also a morphism from the product quasi-Borel space.
So  $f_{k^\sharp}$ is a morphism,
and we can form $(\curry{f_{k^\sharp}}) \colon \RR\to \RR^X$.
So,
\begin{align*}
        q([\curry{f_{k^\sharp}},\,\upsilon])(x)
        & = [\lambda r.\curry{f_{k^\sharp}}(r)(x),\,\upsilon]
        \\
        & = [\lambda r. f_{k^\sharp}(r,x),\,\upsilon]
          = k(x)\text{,}
\end{align*}
and $q$ is surjective, as required.

Finally we show that $\qb{(\Pmonad(\RR))^X}=\{q\circ \alpha~|~
\alpha\in\qb {\Pmonad(\RR^X)}\}$.
We have $(\supseteq)$ since $q$ is a morphism, so it remains to show $(\subseteq)$.
Consider $\beta\in\qb{(\Pmonad(\RR))^X}$.
We must show that $\beta=q\circ\alpha$ for some $\alpha\in \qb{\Pmonad(\RR^X)}$.
By Prop.~\ref{prop:functionspace}, $\beta\in\qb{(\Pmonad(\RR))^X}$
means the uncurried function
$(\uncurry\,\beta)\colon\RR\times X\to \Pmonad (\RR)$ is a morphism.
As above, this morphism corresponds to a kernel
$(\uncurry\,\beta)^\sharp\colon (\RR\times X)\times \sigalg \RR\to [0,1]$.
The randomization lemma (Prop.~\ref{prop:randomization}) gives a measurable function $f_{\beta}\colon \RR\times (\RR\times X)\to \RR$ such that
$(\uncurry\,\beta)^\sharp((r,x),U)=\upsilon\{s \mid f_{\beta}(s,(r,x))\in U\}$.
By Prop.~\ref{prop:adjunction}(1) and the fact that
the $\sigma$-algebra of a product quasi-Borel space
$\RR \times (\RR \times X)$ includes the product $\sigma$-algebras
$\sigalg \RR \otimes \sigalg {\qb {\RR \times X}}$,
this function $f_\beta$ is also a morphism.
Define $\gamma\colon \RR\to (\RR^X)^\RR$ by
${\gamma = \lambda r.\,\lambda s.\,\lambda x.\,f_\beta(s,(r,x))}$.
This is a morphism since we can interpret $\lambda$-calculus in a cartesian closed
category.
Define $\alpha\colon \RR\to \Pmonad(\RR^X)$ by
$\alpha(r)= [\gamma(r),\upsilon]$;
this function is in $\qb{\Pmonad(\RR^X)}$ by Lemma~\ref{lemma:alt-qb-pmonad}.
A direct calculation now gives $\beta=q\circ\alpha$, as required.
\end{proof}
}

\section{De Finetti's theorem}\label{sec:definetti}
De Finetti's theorem \cite{deFinetti37} is one of the foundational results in Bayesian statistics.
It says that every exchangeable sequence of random observations on $\RR$ 
or another well-behaved measurable space can be modeled accurately by the following 
two-step process: first choose a probability 
measure on $\RR$ randomly (according to some distribution on probability measures) 
and then generate a sequence with independent samples from this measure. Limiting observations to values in a well-behaved space like $\RR$ in the theorem is important: Dubins and Freedman 
proved that the theorem fails for a general measurable space~\cite{Dubins1979}.

In this section, we show that a version of de Finetti's theorem holds for all quasi-Borel spaces, not just $\RR$. Our result does not 
contradict Dubins and Freedman's obstruction; probability measures on quasi-Borel spaces may only use $\RR$ as their 
source of randomness, whereas those on measurable spaces are allowed to use
any measurable space for the same purpose. As we will show shortly, this careful choice
of random source lets us generalize key arguments in a proof
of de Finetti's theorem~\cite{Austin-IISC13} to quasi-Borel spaces.

Let $(X,\qb X)$ be a quasi-Borel space and $(\Xn, \qb \Xn)$ 
the product quasi-Borel space $\prod_{i = 1}^n X$ for each positive integer $n$.
Recall that $\Pmonad(X)$ consists of equivalence classes $[\beta,\nu]$ of probability measures $(\beta,\nu)$ on $X$.
For $n \geq 1$, define a morphism $\iidn \colon \Pmonad(X) \to \Pmonad(\Xn)$ by
\[
        \textstyle{\iidn([\beta,\nu])}
        = \textstyle{\left[\left(\prod_{i = 1}^n \beta \circ \iotan\right),\,
        \left(\left(\inviotan\right)_* \bigotimes_{i = 1}^n \nu\right)\right]}
\] 
where $\iotan$ is a measurable isomorphism $\RR \to \prod_{i =1}^n \RR$,
and $\bigotimes_{i = 1}^n \nu$ is the product measure formed by $n$ copies of $\nu$. 
The name $\iidn$ represents `independent and identically distributed'.
Indeed, $\iidn$ transforms a probability measure $(\beta,\nu)$ on $X$ to the measure of the random sequence in $\Xn$ that independently samples from $(\beta,\nu)$.
The function $\iidn$ is a morphism $\Pmonad(X) \to \Pmonad(\Xn)$ because 
it can also be written in terms of the strength of the monad $\Pmonad$.

Write $(\XI,\qb \XI)$ for the countable product $\prod_{i = 1}^\infty X$.
\begin{definition} 
        A probability measure $(\alpha,\mu)$ on $\XI$ is \emph{exchangeable} 
        if for all permutations $\pi$ on positive integers, 
        \shortversion{
        $[\alpha,\mu] = [\alpha_\pi, \mu]$, 
        where $\alpha_\pi(r)_i \defeq \alpha(r)_{\pi(i)}$ for all $r$ and $i$.
        }
        \longversion{
        \[
                [\alpha,\mu] = [\alpha_\pi, \mu]\text,
        \]
        where $\alpha_\pi(r)_i \defeq \alpha(r)_{\pi(i)}$ for all $r \in \RR$ and $i \geq 1$.
        }
\end{definition}

\newcommand{\projn}[1]{(#1)_{1\dots n}}
\newcommand{\projnname}{\projn{-}}
\begin{theorem}[Weak de Finetti for quasi-Borel spaces]
        \label{thm:deFinetti-qbs}
        If $(\alpha,\mu)$ is an exchangeable probability measure on $\XI$,
        then there exists a probability measure $(\beta,\nu)$ in $\Pmonad(\Pmonad(X))$ such that
        for all $n \geq 1$, the measure $(\bind{[\beta,\nu]} \iidn)$ on $\Pmonad(\Xn)$ equals $\Pmonad(\projnname)(\alpha,\mu)$
        when considered as a measure on the 
        product measurable space $(\Xn,\bigotimes_{i = 1}^n\qbtosig {\qb X})$. 
        (Here $\projnname\colon \XI\to X^n$ is $\projn x\defeq (x_1,\dots,x_n)$.)
\end{theorem}
In the theorem, $(\beta,\nu)$ represents a random variable that has a probability measure on $X$
as its value. The theorem says that (every finite prefix of) a sample sequence from $(\alpha,\mu)$ can be generated
by first sampling a probability measure on~$X$ according to $(\beta,\nu)$, then generating
independent $X$-valued samples from the measure, and finally forming a sequence with these samples.

We call the theorem \emph{weak} for two reasons. First, the $\sigma$-algebra 
$\qbtosig {\qb \Xn}$ includes the product $\sigma$-algebra $\bigotimes_{i = 1}^n \qbtosig {\qb X}$,
but we do not know that they are equal; two different probability 
measures in $\Pmonad(\Xn)$ may induce the same measure on $(\Xn,\bigotimes_{i = 1}^n \qbtosig {\qb X})$,
although they always induce different measures on $(\Xn,\qbtosig {\qb \Xn})$. In the theorem,
we equate such measures, which lets us use a standard
technique for proving the equality of measures on product $\sigma$-algebras. Second, we are unable to
construct a version of $\iidn$ for infinite sequences, i.e.\ a morphism $\Pmonad(X) \to \Pmonad(\XI)$
implementing the independent identically-distributed random sequence. The theorem is stated only for finite prefixes. 

The rest of this section provides an overview of our proof of Theorem~\ref{thm:deFinetti-qbs}.
The starting point is to unpack definitions in the theorem, especially those related to quasi-Borel spaces, and to rewrite the statement of the theorem purely in terms of standard measure-theoretic notions.
\begin{lemma}
        \label{lemma:deFinetti-qbs:paraphrase}
        Let $(\alpha,\mu)$ be an exchangeable probability measure on $\XI$. Then,
        the conclusion of Theorem~\ref{thm:deFinetti-qbs} holds if and only if
        there exist a probability 
        measure $\xi \in G(\RR)$, a measurable function
        $k \colon \RR \to G(\RR)$, and $\gamma \in \qb X$ 
        such that for all $n \geq 1$ and all $U_1,\ldots,U_n \in \qbtosig {\qb X}$,
        \begin{multline*}
                \int_{r \in \RR} \left(\prod_{i=1}^n [\alpha(r)_i \in U_i]\right)\, \dd\mu 
                \\
                {}
                = 
                \int_{r \in \RR} \prod_{i = 1}^n \left(\int_{s \in \RR} \left[\gamma(s) \in U_i\right] \dd (k(r))\right)\dd\xi\text.  
        \end{multline*}
        Here 
        we express the domain of integration and the integrated variable explicitly to avoid confusion.
\end{lemma}
\shortversion{
\begin{proof}
        Let $(\alpha,\mu)$ be an exchangeable probability measure on~$\XI$. We 
        unpack definitions in the conclusion of Theorem~\ref{thm:deFinetti-qbs}. 
        The first definition to unpack is the notion of probability measure in $\Pmonad(\Pmonad(X))$. 
        Here are the crucial facts that enable this unpacking. First, for every probability measure 
        $(\beta,\nu)$ on $\Pmonad(X)$, there exist a function $\gamma \colon \RR \to X$ in $\qb X$ 
        and a measurable $k \colon \RR \to G(\RR)$ such that $\beta(r) = [\gamma,k(r)]$
        for all $r \in \RR$. Second, conversely, for a function $\gamma \in \qb X$, 
        a measurable $k \colon \RR \to G(\RR)$, and a probability measure $\nu \in G(\RR)$,
        the function $(\lambda r.\,[\gamma,k(r)],\nu)$ 
        is a probability measure in $\Pmonad(\Pmonad(X))$. Thus, we can look for $(\gamma,k,\nu)$ in 
        the conclusion of the theorem instead of $(\beta,\nu)$.  
        
        The second is the definition of $\bind{[\beta,\nu]}{\iidn}$. Using 
        $(\gamma,k,\nu)$ instead of $(\beta,\nu)$, we find that
        $\bind{[\beta,\nu]}\iidn$ is the measure $[(\prod_{i=1}^n\gamma)\circ \iotan,\;(\inviotan)_*\,(\bind\nu{\lambda r.\,\bigotimes_{i=1}^n k(r)})]$. 
        
        Recall that two measures $p$ and $q$ on the product space 
        $(\Xn,\bigotimes_{i = 1}^n X)$ are 
        equivalent when
        $p(U_1\times\cdots\times U_{n})$ equals $q(U_1\times\cdots\times U_{n})$
        for all $U_1,\ldots, U_n\in \qbtosig{\qb X}$. Thus we must show that
        $(\projnname \circ \alpha)_*\mu)(U_1\times\ldots\times U_{n})$ is equal to $\big((\prod_{i = 1}^n\gamma)_*\,(\bind\nu{\lambda r.\,\bigotimes_{i = 1}^nk(r)})\big)(U_1\times\ldots\times U_{n})$. 
        This equation is equivalent to the one in the statement of the lemma with $\xi= \nu$. 
\end{proof}
}

Thus we just need to show how to construct $\xi$, $k$ and $\gamma$ in Lemma~\ref{lemma:deFinetti-qbs:paraphrase}
from a given exchangeable probability measure $(\alpha,\mu)$ on $\XI$.
Constructing $\xi$ and $\gamma$ is easy:
\[
        \xi \defeq \mu,\qquad
        \gamma \defeq \lambda r.\,\alpha(r)_1\text.
\]
Note that these definitions type-check: $\xi = \mu \in G(\RR)$,
and $\gamma \in \qb X$ because $\alpha \in \qb \XI$ and the first projection
$(-)_1$ is a morphism $\XI \to X$. 

Constructing $k$ is not that easy. We need to use the fact that
$\mu$ is defined over $\RR$, a standard Borel space. This fact itself
holds because all probability measures on quasi-Borel spaces use $\RR$
as their source of randomness. Define measurable functions
$\alpha_e,\alpha_o \colon (\RR,\sigalg \RR) \to (\XI,\qbtosig {\qb \XI})$ by
\[
\alpha_e(r)_i \defeq \alpha(r)_{2i}\quad\text{(even),}\qquad
\alpha_o(r)_i \defeq \alpha(r)_{2i-1}\quad\text{(odd)}\text.
\]
Since $\mu$ is a probability measure on $\RR$,  there exists a measurable function
$k' \colon (\XI,\qbtosig {\qb \XI}) \to (G(\RR),\sigalg {G(\RR)})$,
called a conditional probability kernel, such that for all measurable $f \colon \RR \to \RR$
and $U \in \inv{(\alpha_e)}(\qbtosig {\qb \XI})$,
\begin{equation}
        \label{eqn:deFinetti-qbs:0}
        \int_{r \in U} f(r) \,\dd\mu = \int_{r \in U} \left(\int_\RR f \,\dd((k' \circ \alpha_e)(r))\right) \dd\mu\text.
\end{equation}
Define $k \defeq k' \circ \alpha_e$.

Our $\xi$, $k$ and $\gamma$ satisfy the requirement in Lemma~\ref{lemma:deFinetti-qbs:paraphrase}
because of the following three properties, which follow from exchangeability of
$(\alpha,\mu)$.
\begin{lemma} 
        \label{lemma:deFinetti-qbs:odd}
        For all $n \geq 1$ and all $U_1,\ldots,U_n \in \qbtosig {\qb X}$,
        \[
                \int_{r \in \RR} \left(\prod_{i=1}^n [\alpha(r)_i \in U_i]\right)\, \dd\mu 
                =
                \int_{r \in \RR} \left(\prod_{i=1}^n [\alpha_o(r)_i \in U_i]\right)\, \dd\mu \text.
        \]
\end{lemma}
\begin{proof}
        Consider $n \geq 1$ and $U_1,\ldots,U_n \in \qbtosig {\qb X}$. 
        Pick a permutation $\pi$ on positive integers such that 
        $\pi(i) = 2i-1$ for all integers $1 \leq i \leq n$.
        Then, $[\alpha,\mu] = [\alpha_\pi,\mu]$ 
        by the exchangeability of $(\alpha,\mu)$. Thus
        \begin{align*}
                \int_{r \in \RR} \left(\prod_{i=1}^n [\alpha(r)_i \in U_i]\right) \dd\mu 
                & =
                \int_{r \in \RR} \left(\prod_{i=1}^n [\alpha_\pi(r)_i \in U_i]\right) \dd\mu 
                \text,
        \end{align*}
        from which the statement follows.
\end{proof}
\begin{lemma} 
        \label{lemma:deFinetti-qbs:marginal}
        For all $U \in \qbtosig {\qb X}$ and all $i,j \geq 1$,
        \[
                \int_{s \in \RR} \left[\alpha_o(s)_i \in U\right] \dd(k(r)) 
                = 
                \int_{s \in \RR} \left[\alpha_o(s)_j \in U\right] \dd(k(r))
        \]
        holds for $\mu$-almost all $r \in \RR$.
\end{lemma}
\shortversion{
\begin{proof}
        Consider a measurable set ${U \in \qbtosig {\qb X}}$ and ${i,j \geq 1}$.
        The function ${\lambda r.\;\int_{s \in \RR} \left[\alpha_o(s)_i \in U\right] \dd(k(r))}:%
        {\RR\to \RR}$ is a conditional expectation of the indicator function ${\lambda s.\,[\alpha_o(s)_i \in U]}$
        with respect to the probability measure $\mu$ and the $\sigma$-algebra 
        generated by the measurable function $\alpha_e \colon \RR \to (\XI,\qbtosig {\qb \XI})$.
        By the almost-sure uniqueness of conditional expectation, it suffices to show that
        $\lambda r.\, \int_{s \in \RR} \left[\alpha_o(s)_j \in U\right] \dd(k(r))$
        is also a conditional expectation of $\lambda s.\,[\alpha_o(s)_i \in U]$
        with respect to $\mu$ and $\alpha_e$. Pick a measurable subset $V \in \qbtosig {\qb \XI}$. Then:
        \begin{align*}
                & \int_{r \in \RR} \left[\alpha_e(r) \in V\right] \cdot 
                \left(\int_{s \in \RR} \left[\alpha_o(s)_j \in U\right] \dd(k(r))\right)
                \dd\mu
                \\
                & \qquad \qquad \qquad {} =
                \int_{r \in \RR} \left[\alpha_o(r)_j \in U \wedge \alpha_e(r) \in V\right] \dd\mu
                \\
                & \qquad \qquad \qquad {} =
                \int_{r \in \RR} \left[\alpha_o(r)_i \in U \wedge \alpha_e(r) \in V\right] \dd\mu\text.
        \end{align*}
        The first equation holds because the function 
        $\lambda r.\, \int_{s \in \RR} \left[\alpha_o(s)_j \in U\right] \dd(k(r))$ 
        is a conditional expectation of $\lambda s.\,[\alpha_o(s)_j \in U]$ with respect to $\mu$ and $\alpha_e$. 
        The second equation follows from the exchangeability of $(\alpha,\mu)$. We have just shown that 
        $\lambda r.\, \int_{s \in \RR} \left[\alpha_o(s)_j \in U\right] \dd(k(r))$ is a conditional 
        expectation of $\lambda s.\,[\alpha_o(s)_i \in U]$
        with respect to $\mu$ and $\alpha_e$. 
\end{proof}
}
\begin{lemma} 
        \label{lemma:deFinetti-qbs:independence}
        For all $n \geq 1$ and all $U_1,\ldots,U_n \in \qbtosig {\qb X}$, 
        \begin{multline*}
                \int_{s \in \RR} \left(\prod_{i = 1}^n\left[\alpha_o(s)_i \in U_i\right]\right) \dd(k(r)) 
                \\
                {}
                = 
                \prod_{i = 1}^n \int_{s \in \RR} \left[\alpha_o(s)_i \in U_i\right] \dd(k(r))
        \end{multline*}
        holds for $\mu$-almost all $r \in \RR$. 
\end{lemma} 
\shortversion{
\begin{proof}[Proof notes]
        Use induction on $n \geq 1$. There is nothing to
        prove for the base case $n = 1$. To handle the inductive case, assume that $n > 1$.
        Let $U_1,\ldots,U_n$ be subsets in $\qbtosig {\qb X}$.
        Define a function $\alpha' \colon \RR \to \XI$ as follows:
        \[
                \alpha'(r)_i = 
                \left\{\begin{array}{ll}
                        \alpha_o(r)_i & \mbox{if $1 \leq i \leq n-1$}
                        \\
                        \alpha_e(r)_{i-n+1} & \mbox{otherwise.}
                \end{array}\right.
        \]
        Then, $\alpha'$ is in $\qb \XI$, so that $\alpha'$ is a measurable
        function $(\RR,\sigalg \RR) \to (\XI,\sigalg {\qb \XI})$. Thus there exists a measurable 
        $k'_0 \colon (\XI,\qbtosig {\qb \XI}) \to (G(\RR),\sigalg {G(\RR)})$, the
        conditional probability kernel, such that
        for all measurable functions $f \colon \RR \to \RR$, 
        $\lambda r.\, \int_\RR f \,\dd((k'_0 \circ \alpha')(r))$
        is a conditional expectation of $f$ with respect to $\mu$
        and the $\sigma$-algebra generated by~$\alpha'$.  
        Define $k' \colon \RR \to G(\RR) = k'_0 \circ \alpha'$.
        Then $k'$ is measurable because so are $k'_0$ and $\alpha'$. 
        More importantly, for $\mu$-almost all $r \in \RR$, 
        \begin{equation}
                \int_{s \in \RR} \left[\alpha_o(s)_n \in U_n\right] \dd(k(r))
                =
                \int_{s \in \RR} \left[\alpha_o(s)_n \in U_n\right] \dd(k'(r))\text.
                \label{eqn:deFinetti:3}
        \end{equation}
        The proof of this equality appears in the full version of this paper.
        
        Recall that $k = k_0 \circ \alpha_e$ and $k' = k'_0 \circ \alpha'$ are 
        defined in terms of conditional expectation. Thus, they inherit all the properties
        of conditional expectation. In particular, 
        for $\mu$-almost all $r \in \RR$
        and all measurable $h \colon \RR \to \RR$,
        \begin{align}
                \begin{split}
                & \int_{s \in \RR} \prod_{i = 1}^{n} \left[\alpha_o(s)_i \in U_i\right] \dd(k(r))
                \\
                & \quad {}=
                \int_{s \in \RR} \left(\int_{t \in \RR} \prod_{i = 1}^{n} \left[\alpha_o(t)_i \in U_i\right] \dd(k'(s))\right) 
                \dd(k(r))\text,
                \end{split}
                \label{eqn:deFinetti:4}
                \\[1ex]
                \begin{split}
                & \int_{s \in \RR} \prod_{i = 1}^{n} \left[\alpha_o(s)_i \in U_i\right] \dd(k'(r))
                \\
                & \quad {}=
                \prod_{i = 1}^{n-1} \left[\alpha_o(r)_i \in U_i\right]
                \cdot
                \int_{s \in \RR} \left[\alpha_o(s)_n \in U_n\right] \dd(k'(r))\text,
                \end{split}
                \label{eqn:deFinetti:5}
                \\[1ex]
                \begin{split}
                & \int_{s \in \RR} \left(h(s)
                \cdot \int_{t \in \RR} \left[\alpha_o(t)_n \in U_n\right] \dd(k(s))\right)
                \dd(k(r))
                \\
                & \quad {}=
                \left(\int_{t \in \RR} \left[\alpha_o(t)_n \in U_n\right] \dd(k(r))\right)
                \cdot 
                \left(\int_{s \in \RR} h(s)\, \dd(k(r))\right)\text.
                \end{split}
                \label{eqn:deFinetti:6}
        \end{align}
        Using the assumption \eqref{eqn:deFinetti:3}
        and the properties \eqref{eqn:deFinetti:4}, \eqref{eqn:deFinetti:5} and \eqref{eqn:deFinetti:6},
        we complete the proof of the inductive case as follows: 
        for all subsets $V \in \inv{(\alpha_e)}(\qbtosig {\qb \XI})$,
        \begin{align*}
                &
                \int_{r \in V}
                \int_{s \in \RR} 
                \prod_{i = 1}^{n} \left[\alpha_o(s)_i \in U_i\right] \dd(k(r))\,\dd\mu
                \\
                &
                {} =
                \int_{r \in V}
                \int_{s \in \RR} 
                \int_{t \in \RR} \prod_{i = 1}^{n} \left[\alpha_o(t)_i \in U_i\right] 
                \dd(k'(s))\,\dd(k(r))\,\dd\mu
                \\
                &
                {} =
                \int_{r \in V}
                \int_{s \in \RR}
                \prod_{i = 1}^{n-1} \left[\alpha_o(s)_i \in U_i\right]
                \\*
                &
                \phantom{{} =
                \int_{r \in V}
                \int_{s \in \RR}}
                {} \cdot
                \int_{t\in\RR} \left[\alpha_o(t)_n \in U_n\right] \dd(k'(s))\,
                \dd(k(r))\,
                \dd\mu
                \\
                &
                {} =
                \int_{r \in V}
                \int_{s \in \RR}
                \prod_{i = 1}^{n-1} \left[\alpha_o(s)_i \in U_i\right]
                \\*
                &
                \phantom{
                {} =
                \int_{r \in V}
                \int_{s \in \RR}}
                {}\cdot
                \int_{t \in \RR} \left[\alpha_o(t)_n \in U_n\right] \dd(k(s))\,\dd(k(r))\,\dd\mu
                \\
                &
                {} =
                \int_{r \in V}
                \left(\int_{t \in \RR} \left[\alpha_o(t)_n \in U_n\right] \dd(k(r))\right)
                \\*
                &
                \phantom{
                {} =
                \int_{r \in V}}
                {}
                \cdot
                \left(\int_{s \in \RR} \prod_{i = 1}^{n-1} \left[\alpha_o(s)_i \in U_i\right] \dd(k(r))\right)\,
                \dd\mu
                \\
                &
                {} =
                \int_{r \in V}
                \prod_{i = 1}^n \int_{s \in \RR} \left[\alpha_o(s)_i \in U_i\right] 
                \dd(k(r))\,\dd\mu\text.
        \end{align*}
        The first and the second equalities hold because of \eqref{eqn:deFinetti:4} and \eqref{eqn:deFinetti:5}.
        The third equality uses \eqref{eqn:deFinetti:3}, and the fourth the equality 
        in \eqref{eqn:deFinetti:6}. The fifth follows from the induction hypothesis. Our derivation
        implies that both
                $\lambda r.\,\int_{s \in \RR} \prod_{i = 1}^{n} \left[\alpha_o(s)_i \in U_i\right] \dd(k(r))$
                and
                $\lambda r.\,\prod_{i = 1}^n \int_{s \in \RR} \left[\alpha_o(s)_i \in U_i\right] \dd(k(r))$
        are conditional expectations of the same function with respect to $\mu$ and the same $\sigma$-algebra.
        So, they are equal for $\mu$-almost all inputs~$r$.
\end{proof}
}
The following calculation combines these lemmas and shows that $\xi$, $k$ and $\gamma$ 
satisfy the requirement in Lemma~\ref{lemma:deFinetti-qbs:paraphrase}: 
\begin{align*} 
        & \int_{r \in \RR} \prod_{i = 1}^n \left[\alpha(r)_i \in U_i\right] \dd\mu
        \\ 
        &
        = \int_{r \in \RR} \prod_{i = 1}^n \left[\alpha_o(r)_i \in U_i\right] \dd\mu 
        & \mbox{Lem.~\ref{lemma:deFinetti-qbs:odd}}
        \\ 
        &
        = \int_{r \in \RR} \left(\int_{s \in \RR} \prod_{i = 1}^n \left[\alpha_o(s)_i \in U_i\right] \dd(k(r))\right)\dd\mu 
        & \mbox{Eq.~\eqref{eqn:deFinetti-qbs:0}}
        \\
        &
        = \int_{r \in \RR} \prod_{i = 1}^n \left(\int_{s \in \RR} \left[\alpha_o(s)_i \in U_i\right] \dd(k(r))\right) \dd\mu 
        & \mbox{Lem.~\ref{lemma:deFinetti-qbs:independence}}
        \\ 
        &
        = \int_{r \in \RR} \prod_{i = 1}^n \left(\int_{s \in \RR} \left[\alpha_o(s)_1 \in U_i\right] \dd(k(r))\right) \dd\mu 
        & \mbox{Lem.~\ref{lemma:deFinetti-qbs:marginal}}
        \\ 
        &
        = \int_{r \in \RR} \prod_{i = 1}^n \left(\int_{s \in \RR} \left[\gamma(s) \in U_i\right] \dd(k(r))\right) \dd\xi
        & \mbox{Def. of $\gamma,\xi$}\text.  
\end{align*}
This concludes our proof outline for Theorem~\ref{thm:deFinetti-qbs}.

\section{Related work}\label{sec:related}

\subsection{Quasi-topological spaces and categories of functors}

Our development of a cartesian closed category from measurable spaces mirrors the development of cartesian closed categories of topological spaces
over the years. 

For example, quasi-Borel spaces are reminiscent of \emph{subsequential spaces}~\cite{johnstone-topological-topos}: a set $X$ together with a collection of functions $Q\subseteq {[\NN\cup \{\infty\}\to X]}$ satisfying some conditions. 
The functions in  $Q$ are thought of as convergent sequences.
Another notion of generalized topological space is \emph{C-space}~\cite{xu-escardo}: a set~$X$ together with a collection $Q\subseteq [2^\NN\to X]$
of `probes' satisfying some conditions;
this is a variation on Spanier's early notion of \emph{quasi-topological space}~\cite{spanier:quasitopologies}.
Another reminiscent notion in the context of differential geometry is a \emph{diffeological space}~\cite{bh-convenient}:
a set $X$ together with a set $Q_U\subseteq [U\to X]$ of `plots' for each open subset $U$ of $\RR^n$ satisfying some conditions.
These examples all form cartesian closed categories. 

A common pattern is that these spaces can be understood as extensional (concrete) sheaves 
on an established category of spaces. 
Let $\SMeas$ be the category of standard Borel spaces and measurable functions.
There is a functor
\shortversion{
  $J\colon \QBS\to[\op\SMeas,\Set]$
}
\longversion{
\[
  J\colon \QBS\to[\op\SMeas,\Set]
\]
}
with $\big(J(X,\qb X))(Y,\sigalg Y\big)\defeq \QBS\big((Y,\sigtoqb{\sigalg Y}),(X,\qb X)\big)$,
which is full and faithful by Prop.~\ref{prop:adjunction}(2).
We can characterize those functors that arise in this way.
\begin{proposition}\label{prop:extensionalpresheaf}
  Let $F\colon\op\SMeas\to\Set$ be a functor.
  The following are equivalent:
  \begin{itemize}
  \item $F$ is naturally isomorphic to $J(X,\qb X)$, for some quasi-Borel space $(X,\qb X)$;
  \item $F$ preserves countable products and $F$ is extensional: the functions $i_{(X,\sigalg X)}\colon F(X,\sigalg X)\to\Set(X,F(1))$ are injective, where $(i_{(X,\sigalg X)}(\xi))(x)=(F(\ulcorner x\urcorner))(\xi)$, and we consider $x\in X$ as a function $\ulcorner x\urcorner \colon 1\to X$. 
  \end{itemize}
\end{proposition}

There are similar characterizations of subsequential spaces~\cite{johnstone-topological-topos}, quasi-topological spaces~\cite{dubuc-concrete-quasitopoi} and diffeological spaces~\cite{bh-convenient}. Prop.~\ref{prop:extensionalpresheaf}
is an instance of a general pattern (e.g.~\cite{bh-convenient,dubuc-concrete-quasitopoi});
but that is not to say that the definition of quasi-Borel space (Def.~\ref{def:qbs})
arises automatically.
The method of 
extensional presheaves also arises in other models of computation
such as finiteness spaces~\cite{ehrhard-extensional} and realizability models~\cite{rosolini-streicher}. 
This work appears to be the first application to probability theory, 
although via Prop.~\ref{prop:extensionalpresheaf} there are connections 
to Simpson's probability sheaves~\cite{simpson-cippmi}.

The characterization of Prop.~\ref{prop:extensionalpresheaf} gives a canonical categorical status to quasi-Borel spaces.
It also connects with our earlier work~\cite{statonyangheunenkammarwood:higherorder}, which used the cartesian closed category of countable-product-preserving functors
in $[\op\SMeas, \Set]$.
Quasi-Borel spaces have several advantages over this functor category. 
For one thing, they are more concrete, leading
to better intuitions for their constructions. For example, measures in~\cite{statonyangheunenkammarwood:higherorder} are built abstractly from left Kan extensions, whereas for quasi-Borel spaces they have a straightforward concrete definition (Def.~\ref{def:probabilitymeasure}). 
For another thing, in contrast to the functor category in~\cite{statonyangheunenkammarwood:higherorder}, quasi-Borel spaces form a well-pointed category:
if two morphisms $(X,\qb X)\to (Y,\qb Y)$ are different
then they disagree on some point in $X$.
From the perspective of semantics of programming languages, 
where terms in context $\Gamma\vdash t :A$ are interpreted as morphisms
$\denot t \colon \denot \Gamma\to\denot A$, well-pointedness is a crucial property.
It says that if two open terms
are different, $\denot t\neq \denot u:\denot \Gamma\to\denot A$, 
then there is a ground context $\mathcal C \colon 1\to\denot \Gamma$ that 
distinguishes them: $\denot{\mathcal C[t]}\neq \denot{\mathcal C[u]}:1\to \denot A$. 

Quasi-Borel spaces add objects to make the category of measurable spaces cartesian closed. Another interesting future direction is to add morphisms to make more objects isomorphic, and so find a cartesian closed subcategory~\cite{steenrod:convenient}.


\subsection{Domains and valuations}
In this paper our starting point has been the standard foundation for probability theory, 
based on $\sigma$-algebras and probability measures. 
An alternative foundation for probability is based on topologies and valuations. 
An advantage of our starting point is that we can reference the canon of work on
probability theory. Having said this, an advantage to the approach based on valuations
is that it is related to domain theoretic methods, which have already been used to 
give semantics to programming languages. 

Jones and Plotkin~\cite{jones-plotkin} showed that valuations form a monad which is 
analogous to our probability monad. However, there is considerable debate
about which cartesian closed category this monad should be based on~(e.g.~\cite{jung-tix,gl-qrb-domains}).
For a discussion of the concerns in the context of programming languages, see e.g.~\cite{escardo-high-type-prob-testing}. 
One recent proposal is to use Girard's probabilistic coherence spaces~\cite{etp-pcoh}.
Another is to use a topological domain theory as a cartesian closed category for analysis and probability~(\cite{bss-convenient-domains,pape-streicher,huang-morrisett}). 

Concerns about probabilistic powerdomains have led instead to domains of random variables~(e.g.~\cite{mislove-randvar,barker-monad,scott-stochastic}). 
We cannot yet connect formally with this work, but there are many intuitive links. For example, our measures on quasi-Borel spaces (Def.~\ref{def:probabilitymeasure}) are reminiscent of continuous random variables on a dcpo.

An additional advantage of a domain theoretic approach is that it naturally  
supports recursion. We are currently investigating a notion of `ordered quasi-Borel
space', by enriching Prop.~\ref{prop:extensionalpresheaf} over dcpo's.


\subsection{Other related work}
Our work is related to two recent semantic studies on probabilistic 
programming languages. The first is Borgstr\"om et al.'s \emph{operational} (not denotational as
in this paper) semantics for
a higher-order probabilistic programming language with continuous 
distributions~\cite{blgs-lambda-prob-untyped},
which has been used to justify a basic inference algorithm for the language. 
Recently, Culpepper and Cobb refined 
this operational approach using logical relations~\cite{Culpepper-esop17}. The second study is Freer 
and Roy's results on a computable variant of de Finetti's theorem and its implication 
on exchangeable random processes implemented in
higher-order probabilistic programming languages~\cite{FreerR12}. One interesting future direction is 
to revisit the results about logical relations and computability in these studies with quasi-Borel spaces,
and to see whether they can be extended to spaces other than standard Borel spaces.


\section{Conclusion}\label{sec:conclusion}

We have shown that quasi-Borel spaces (\S\ref{sec:quasiborel})
support higher-order functions (\S\ref{sec:structure}) 
as well as spaces of probability measures (\S\ref{sec:giry}). 
We have illustrated the power of this new formalism by
giving a semantic analysis of Bayesian regression (\S\ref{sec:example}),
by rephrasing the randomization lemma as a quotient-space construction (\S\ref{sec:functions}),
and by showing that it supports de Finetti's theorem (\S\ref{sec:definetti}).


\shortversion{
\section*{Acknowledgment}
We thank Radha Jagadeesan and Dexter Kozen for encouraging us to think about a well-pointed cartesian closed category for probability theory, Vincent Danos and Dan Roy for nudging us to work on de Finetti's theorem, Mike Mislove for discussions of quasi-Borel spaces, and Martin Escard\'o for explaining C-spaces, and Alex Simpson for detailed report with many suggestions. This research was supported by a Royal Society Research Fellowship and EPSRC grants EP/L002388/2 and EP/N007387/1, and also by Institute for Information \& communications Technology Promotion (IITP) grant funded by the Korea government (MSIP) (No.R0190-16-2011, Development of Vulnerability Discovery Technologies for IoT Software Security). 
}
\longversion{
\section*{Acknowledgment}
We thank Radha Jagadeesan and Dexter Kozen for encouraging us to think about a well-pointed cartesian closed category for probability theory, Vincent Danos and Dan Roy for nudging us to work on de Finetti's theorem, Mike Mislove for discussions of quasi-Borel spaces, Martin Escard\'o for explaining C-spaces, and Alex Simpson for detailed report with many suggestions. This research was supported by a Royal Society Research Fellowship and EPSRC grants EP/L002388/2 and EP/N007387/1, and also by an Institute for Information \& communications Technology Promotion (IITP) grant funded by the Korea government (MSIP) (No.R0190-16-2011, Development of Vulnerability Discovery Technologies for IoT Software Security). 
}


\bibliographystyle{IEEEtranS}
\bibliography{lics2017}

\begin{thebibliography}{10}
\providecommand{\url}[1]{#1}
\csname url@samestyle\endcsname
\providecommand{\newblock}{\relax}
\providecommand{\bibinfo}[2]{#2}
\providecommand{\BIBentrySTDinterwordspacing}{\spaceskip=0pt\relax}
\providecommand{\BIBentryALTinterwordstretchfactor}{4}
\providecommand{\BIBentryALTinterwordspacing}{\spaceskip=\fontdimen2\font plus
\BIBentryALTinterwordstretchfactor\fontdimen3\font minus
  \fontdimen4\font\relax}
\providecommand{\BIBforeignlanguage}[2]{{%
\expandafter\ifx\csname l@#1\endcsname\relax
\typeout{** WARNING: IEEEtranS.bst: No hyphenation pattern has been}%
\typeout{** loaded for the language `#1'. Using the pattern for}%
\typeout{** the default language instead.}%
\else
\language=\csname l@#1\endcsname
\fi
#2}}
\providecommand{\BIBdecl}{\relax}
\BIBdecl

\bibitem{aumann:functionspaces}
R.~J. Aumann, ``Borel structures for function spaces,'' \emph{Illinois Journal
  of Mathematics}, vol.~5, pp. 614--630, 1961.

\bibitem{Austin-IISC13}
\BIBentryALTinterwordspacing
T.~Austin, ``Exchangeable random arrays,'' 2013. [Online]. Available:
  \url{https://cims.nyu.edu/\~tim/ExchnotesforIISc.pdf}
\BIBentrySTDinterwordspacing

\bibitem{bh-convenient}
J.~C. Baez and A.~E. Hoffnung, ``Convenient categories of smooth spaces,''
  \emph{Trans. Amer. Math. Soc.}, vol. 363, 2011.

\bibitem{barker-monad}
T.~Barker, ``A monad for randomized algorithms,'' in \emph{Proc.~MFPS}, 2016,
  pp. 47--62.

\bibitem{bss-convenient-domains}
I.~Battenfeld, M.~Schr\"oder, and A.~Simpson, ``A convenient category of
  domains,'' ser. ENTCS, vol. 172, 2007.

\bibitem{blgs-lambda-prob-untyped}
J.~Borgstr\"om, U.~{Dal Lago}, A.~D. Gordon, and M.~Szymczak, ``A
  lambda-calculus foundation for universal probabilistic programming,'' in
  \emph{Proc.~ICFP}, 2016, pp. 33--46.

\bibitem{Culpepper-esop17}
R.~Culpepper and A.~Cobb, ``Contextual equivalence for probabilistic programs
  with continuous random variables and scoring,'' in \emph{Proc.~ESOP}, 2017.

\bibitem{deFinetti37}
B.~de~Finetti, ``La pr\'evision : ses lois logiques, ses sources subjectives,''
  \emph{Annales de l'institut Henri Poincar\'e}, vol.~7, 1937.

\bibitem{Dubins1979}
L.~E. Dubins and D.~A. Freedman, ``Exchangeable processes need not be mixtures
  of independent, identically distributed random variables,''
  \emph{Z.~Angew.~Math.~Mech.}, vol.~48, 1979.

\bibitem{dubuc-concrete-quasitopoi}
E.~J. Dubuc, ``Concrete quasitopoi,'' in \emph{Applications of sheaves}, ser.
  Lect.~Notes Math.\hskip 1em plus 0.5em minus 0.4em\relax Springer, 1977, vol.
  753, pp. 239--254.

\bibitem{ehrhard-extensional}
T.~Ehrhard, ``On finiteness spaces and extensional presheaves over the
  {L}awvere theory of polynomials,'' \emph{J.~Pure Appl.~Algebra}, 2007, to
  appear.

\bibitem{etp-pcoh}
T.~Ehrhard, C.~Tasson, and M.~Pagani, ``Probabilistic coherence spaces are
  fully abstract for probabilistic {P}{C}{F},'' in \emph{POPL 2014}.

\bibitem{escardo-high-type-prob-testing}
M.~H. Escard\'o, ``Semi-decidability of may, must and probabilistic testing in
  a higher-type setting,'' in \emph{Proc.~MFPS}, 2009.

\bibitem{FreerR12}
C.~E. Freer and D.~M. Roy, ``Computable de {F}inetti measures,'' \emph{Ann.
  Pure Appl. Logic}, vol. 163, no.~5, pp. 530--546, 2012.

\bibitem{giry:monad}
M.~Giry, ``A categorical approach to probability theory,'' in \emph{Categorical
  Aspects of Topology and Analysis}, 1982, pp. 68--85.

\bibitem{goodman_uai_2008}
N.~Goodman, V.~Mansinghka, D.~M. Roy, K.~Bonawitz, and J.~B. Tenenbaum,
  ``{Church: a language for generative models},'' in \emph{UAI}, 2008.

\bibitem{gl-qrb-domains}
J.~Goubault-Larrecq, ``$\omega${Q}{R}{B}-domains and the probabilistic
  powerdomain,'' in \emph{Proc.~LICS}, 2010, pp. 352--361.

\bibitem{HewittSavage55}
E.~Hewitt and L.~J. Savage, ``Symmetric measures on cartesian products,''
  \emph{Trans. Amer. Math. Soc}, vol.~80, pp. 470--501, 1955.

\bibitem{huang-morrisett}
\BIBentryALTinterwordspacing
D.~Huang and G.~Morrisett, ``An application of computable distributions to the
  semantics of probabilistic programs: part~2.'' [Online]. Available:
  \url{http://pps2017.soic.indiana.edu/files/2016/12/comp-dist-sem.pdf}
\BIBentrySTDinterwordspacing

\bibitem{johnstone-topological-topos}
P.~Johnstone, ``On a topological topos,'' \emph{Proc.~London Math.~Soc.},
  vol.~3, no.~38, pp. 237--271, 1979.

\bibitem{jones-plotkin}
C.~Jones and G.~D. Plotkin, ``A probabilistic powerdomain of evaluations,'' in
  \emph{Proc.~LICS}, 1989, pp. 186--195.

\bibitem{jung-tix}
A.~Jung and R.~Tix, ``The troublesome probabilistic powerdomain,'' ser. ENTCS,
  vol.~13, 1998, pp. 70--91.

\bibitem{kallenberg}
O.~Kallenberg, \emph{Foundations of Modern Probability}, 2nd~ed.\hskip 1em plus
  0.5em minus 0.4em\relax Springer, 2002.

\bibitem{Mansinghka-venture14}
V.~K. Mansinghka, D.~Selsam, and Y.~N. Perov, ``Venture: a higher-order
  probabilistic programming platform with programmable inference,''
  \emph{arXiv:1404.0099}, 2014.

\bibitem{mislove-randvar}
M.~W. Mislove, ``Anatomy of a domain of continuous random variables~{I},''
  \emph{Theor.~Comput.~Sci.}, vol. 546, pp. 176--187, 2014.

\bibitem{moggi-monads}
E.~Moggi, ``Notions of computation and monads,'' \emph{Inf.~Comput.}, vol.~93,
  no.~1, pp. 55--92, 1991.

\bibitem{pape-streicher}
\BIBentryALTinterwordspacing
M.~Pape and T.~Streicher, ``Computability in basic quantum mechanics,'' 2016.
  [Online]. Available: \url{https://arxiv.org/abs/1610.09209}
\BIBentrySTDinterwordspacing

\bibitem{Preston-Borel08}
\BIBentryALTinterwordspacing
C.~Preston, ``Some notes on standard {B}orel and related spaces,'' 2008.
  [Online]. Available: \url{https://arxiv.org/abs/0809.3066}
\BIBentrySTDinterwordspacing

\bibitem{defunctionalization}
J.~C. Reynolds, ``Definitional interpreters for higher-order programming
  languages,'' in \emph{Proc.~ACM Conference}, 1972.

\bibitem{rosolini-streicher}
G.~Rosolini and T.~Streicher, ``Comparing models of higher type computation,''
  ser. ENTCS, vol.~23, 1999.

\bibitem{scott-stochastic}
D.~S. Scott, ``Stochastic $\lambda$-calculi: An extended abstract,''
  \emph{J.~Applied~Logic}, vol.~12, no.~3, pp. 369--376, 2014.

\bibitem{simpson-cippmi}
\BIBentryALTinterwordspacing
A.~Simpson, ``Probability sheaves,'' 2016, {C}IPPMI. [Online]. Available:
  \url{https://synapse.math.univ-toulouse.fr/index.php/s/QWrxKeXn31mN3gz}
\BIBentrySTDinterwordspacing

\bibitem{spanier:quasitopologies}
E.~Spanier, ``Quasi-topologies,'' \emph{Duke Math.~J.}, pp. 1--14, 1963.

\bibitem{statonyangheunenkammarwood:higherorder}
S.~Staton, H.~Yang, C.~Heunen, O.~Kammar, and F.~Wood, ``Semantics for
  probabilistic programming: higher-order functions, continuous distributions,
  and soft constraints,'' in \emph{LICS}, 2016.

\bibitem{steenrod:convenient}
N.~E. Steenrod, ``A convenient category of topological spaces,'' \emph{Michigan
  Mathematical Journal}, vol.~14, pp. 133--152, 1967.

\bibitem{street-monads}
R.~Street, ``The formal theory of monads,'' \emph{J. Pure Appl. Algebra},
  vol.~2, pp. 149--168, 1972.

\bibitem{wood-aistats-2014}
F.~Wood, J.~W. van~de Meent, and V.~Mansinghka, ``A new approach to
  probabilistic programming inference,'' in \emph{Proc.~AISTATS}, 2014.

\bibitem{xu-escardo}
C.~Xu and M.~{Escard\'o}, ``A constructive model of uniform continuity,'' in
  \emph{Proc.~TLCA}, 2013.

\end{thebibliography}
\end{document}